\documentclass{winnower}

\usepackage{amsmath}
\usepackage{amssymb}
\usepackage{amsfonts}
\usepackage{algorithmic}
\usepackage{textcomp}
\usepackage{graphicx}
\usepackage{multicol}
\usepackage{times}
\usepackage{bbm}
\usepackage{dsfont}
\usepackage{color}
\usepackage[scr]{rsfso}
\usepackage{mathtools}
\usepackage{soul}

\graphicspath{{./figures/}}

\ifdefined\ieeestyle
 \usepackage[caption=false]{subfig}
 \usepackage[capitalise,nameinlink]{cleveref}
 
 \newtheorem{theorem}{Theorem}[section]
 \newtheorem{prop}{Proposition}[section]
 \newtheorem{lemma}{Lemma}[section]
 \newtheorem{assumption}{Assumption}[section]
 \newtheorem{problem}{Problem}[section]
 \newtheorem{remark}{Remark}[section]
 \newtheorem{definition}{Definition}[section]
 \newtheorem{example}{Example}[section]
 \newtheorem{condition}{Condition}
\else
 \ifdefined\corlstyle
  \usepackage[caption=false,justification=centering]{subfig}
  \usepackage{caption}
  \usepackage[capitalise,nameinlink]{cleveref}
  \theoremstyle{plain}
  \newtheorem{theorem}{Theorem}[section]
  \newtheorem{prop}{Proposition}[section]
  \newtheorem{lemma}{Lemma}[section]
  \newtheorem{assumption}{Assumption}[section]
  \newtheorem{problem}{Problem}[section]
  \newtheorem{remark}{Remark}[section]
  
  \theoremstyle{definition}
  \newtheorem{definition}{Definition}[section]
  \newtheorem{example}{Example}[section]
  \newtheorem{condition}{Condition}
 \else
  \usepackage[utf8]{inputenc} 
  \usepackage[T1]{fontenc}    
  \usepackage{amsthm}
  \usepackage[numbers,sort&compress]{natbib}
  \usepackage[justification=centering]{subfig}
  \usepackage{caption}
  \usepackage{hyperref}
  \usepackage[capitalise,nameinlink]{cleveref}
  
  \theoremstyle{plain}
  \newtheorem{theorem}{Theorem}[section]
  
  \newtheorem{lemma}{Lemma}[section]
  \newtheorem{assumption}{Assumption}[section]

  \theoremstyle{definition}
  \newtheorem{definition}{Definition}[section]
  \newtheorem{example}{Example}[section]
  
 \fi
\fi

\crefname{assumption}{Assumption}{Assumptions}

\makeatletter
\newcommand{\printfnsymbol}[1]{%
  \textsuperscript{\@fnsymbol{#1}}%
}
\makeatother

\DeclareMathOperator*{\argmin}{arg\,min}

\DeclarePairedDelimiter\abs{\lvert}{\rvert}%
\DeclarePairedDelimiter\norm{\lVert}{\rVert}%

\makeatletter
\let\oldabs\abs
\def\abs{\@ifstar{\oldabs}{\oldabs*}}
\let\oldnorm\norm
\def\norm{\@ifstar{\oldnorm}{\oldnorm*}}
\makeatother

\newcommand{\ogamma}{\overline{\gamma}}
\newcommand{\oalpha}{\overline{\alpha}}
\newcommand{\ualpha}{\underline{\alpha}}
\newcommand{\olambda}{\overline{\lambda}}
\newcommand{\ulambda}{\underline{\lambda}}

\newcommand{\usigma}{\underline{\sigma}}
\newcommand\Sym[1]{\left[#1\right]_{\mathbb{S}}}

\newcommand{\pdv}[2]{\frac{\partial #1}{\partial #2}}

\usepackage{tikz}
\usetikzlibrary{shapes,arrows.meta}
\usetikzlibrary{positioning}
\usepackage{wrapfig}

\newcommand{\ellone}{\texorpdfstring{$\mathcal{L}_1$}{L1}}
\newcommand{\rellone}{\texorpdfstring{$\mathcal{CL}_1$}{CL1}}
\newcommand{\gp}{\texorpdfstring{$\mathcal{GP}$ }{GP }}

\title{Contraction $\mathcal{L}_1$-Adaptive Control using Gaussian Processes}

\author{Aditya Gahlawat\thanks{These authors contributed equally for this work}}
\author{Arun Lakshmanan\printfnsymbol{1}}
\author{Lin Song}
\author{Andrew Patterson}
\author{Zhuohuan Wu}
\author{Naira Hovakimyan}
\affil{Mechanical Science and Engineering, University of Illinois at Urbana Champaign, \texttt{\{gahlawat,lakshma2,linsong2,appatte2,zw24,nhovakim\}@illinois.edu}}
\author{Evangelos Theodorou}
\affil{School of Aerospace Engineering, Georgia Institute of Technology, \texttt{evangelos.theodorou@gatech.edu}}

\date{}

\begin{document}
\maketitle

\begin{abstract}
We present $\mathcal{CL}_1$-$\mathcal{GP}$, a control framework that enables safe \textit{simultaneous} learning and control for systems subject to uncertainties.
The two main constituents are contraction theory-based $\mathcal{L}_1$ ($\mathcal{CL}_1$) control and Bayesian learning in the form of Gaussian process (GP) regression.
The $\mathcal{CL}_1$ controller ensures that control objectives are met   while providing safety certificates.
Furthermore, $\mathcal{CL}_1$-$\mathcal{GP}$ incorporates any available data into a GP model of uncertainties,  which improves performance and enables the motion planner to achieve optimality \textit{safely}. This way, the safe operation of the system is \textit{always} guaranteed, even during the learning transients. We provide a few illustrative examples for the safe learning and control of planar quadrotor systems in a variety of environments.
\end{abstract}


\section{Introduction}

A majority of planning algorithms based on model predictive control (MPC) and model-based reinforcement learning (MBRL) compute optimal control sequences using a nominal or learned system model. However, models have inaccuracies and the robot may behave sub-optimally. In the worst cases, the system will become unstable or collide with obstacles. These model inaccuracies have especially serious consequences for safety-critical systems \cite{knight2002safety}. Machine learning (ML) algorithms have been proven to be potent tools for learning complex and accurate models in robotics \cite{recht2019tour}, improving performance. However, the robot's safety during the learning transients is not always guaranteed. For instance, a robot may enter unsafe regions while collecting data because it does not take into account the model inaccuracies.

Control-theoretic approaches that offer safety certificates based on Lyapunov functions and robust control invariant sets are gaining popularity \cite{perkins2002lyapunov,berkenkamp2017safe,chow2018lyapunov} in the context of safe robot learning. Many recent safe-learning examples establish the notion of asymptotic stability with control-theoretic tools \cite[Chapter 3]{khalil2014nonlinear}.  Although critically important, asymptotic stability by itself is not sufficient for the safe operation of robots. Safety must be guaranteed during the learning process with transient bounds. Techniques like uncertainty propagation have been proposed to characterize transient performance using learned statistical models \cite{hewing2019cautious}. However, methods based on uncertainty propagation are often approximate, computationally expensive for planning, and do not provide apriori certificates of safety.

\subsection{Our Contributions}

We propose a learning-based control framework using robust adaptive control theory for nonlinear systems that ensures \textit{improvement of optimality and performance while simultaneously guaranteeing safety}  which we refer to as $\mathcal{CL}_1$-$\mathcal{GP}$ control. The safety guarantees are composed of apriori computable transient performance bounds and robustness margins. We rely on Bayesian learning in the form of GP regression to learn the state and time-dependent model uncertainties from noisy measurements. We use the predictive distribution provided by GP learning to compute high-probability error bounds for the estimated uncertainties \cite{lederer2019uniform}. These estimates are then incorporated within the $\mathcal{CL}_1$ robust adaptive control framework recently presented in \cite{lakshmanan2020safe}. Our $\mathcal{CL}_1$-$\mathcal{GP}$ control framework is planner-agnostic and is designed to work with any planner capable of generating desired state and control trajectories using the known (learned or nominal) model. This feature enables the framework to be used in conjunction with many popular planning algorithms such as differential dynamic programming \cite{tassa2012synthesis}, model predictive path integral control \cite{williams2018information}, and sampling-based planners \cite{lavalle2001randomized}, among many others \cite{cichella2017optimal,howell2019altro}.

A critical feature of the $\mathcal{CL}_1$-$\mathcal{GP}$ framework is that it enables MBRL algorithms to achieve optimality as learning progresses but the safety is guaranteed \textit{at all times} regardless of the quality of the learned model. We define safety using the performance bounds and the robustness margins associated with the controller. The performance bounds quantify how far the system trajectory may deviate from the desired trajectory based on the amount of unmodeled uncertainty. The robustness of the controlled system is a function of the available sensors, actuators, and computational hardware. The $\mathcal{CL}_1$ controller provides a sensible approach to balance the trade-off between performance and robustness requirements for safe navigation. However, this trade-off implies that for a specification of robustness margins there is limit on how tight the performance bounds can get. Our framework addresses this problem by using model learning to reduce the effect of the uncertainty which results in tighter performance bounds than would be possible with \cite{lakshmanan2020safe} alone. Moreover, the improved model knowledge and the tighter performance bounds are then incorporated into the planner used by the MBRL algorithm to generate more optimal but still \textit{safe} trajectories as shown in ~\Cref{fig:intro}.

\begin{figure}[t]
    \centering
    \subfloat[]{\label{fig:intro_wide}\includegraphics[width=0.30\columnwidth]{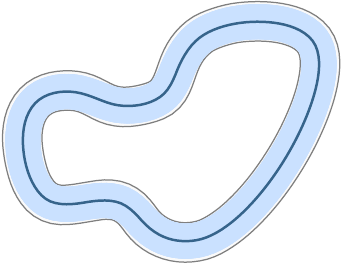}}
    \subfloat[]{\label{fig:intro_medium}\includegraphics[width=0.30\columnwidth]{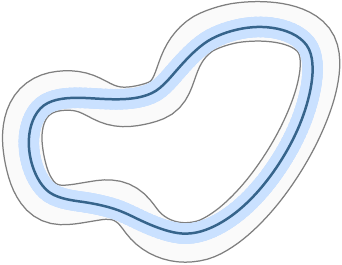}}
    \subfloat[]{\label{fig:intro_narrow}\includegraphics[width=0.30\columnwidth]{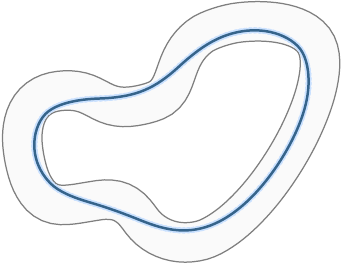}}
    \setlength{\belowcaptionskip}{-10pt}
    \caption{Consider a vehicle traversing a race track with some nominal model knowledge. Depending on the uncertainty and robustness requirements the $\mathcal{CL}_1$-$\mathcal{GP}$ framework will guarantee the performance bounds as denoted by the blue tube around the planned trajectory (a). As learning progresses, the performance bounds become tighter and the learned model can be incorporated into the planner to generate more optimal trajectories (b, c). }
    \label{fig:intro}
\end{figure}

\subsection{Related Work}

Robust MPC problem formulations consider safe planning and control under uncertainties that can be further divided into Min-max MPC~\cite{magni2001receding,raimondo2009min,wang2017adaptive}
or tube-based MPC~\cite{rakovic2016elastic,lopez2019dynamic}.  Min-max approaches plan considering the worst-case disturbance, which might render them overly conservative or even infeasible over long horizons. Tube-based MPC addresses this issue by using an ancillary controller to attenuate disturbances so that the robot stays inside of a tube around the planned trajectory. With the exception of~\cite{lopez2019dynamic}, such methods assume the existence of stabilizing ancillary controllers.  This assumption is limiting, since the design of such controllers cannot be assumed unless further assumptions (e.g. feedback linearizable, strict-feedback) on the dynamics are made.  Learning-based MPC (LBMPC) frameworks have been proposed to reduce conservatism. The LBMPC frameworks use measured data to improve models  and increase optimality. Safe LBMPC methods can be found in~\cite{aswani2013provably} and~\cite{wabersich2018linear} for nominal linear models. Frameworks that use nonlinear nominal models with LBMPC are found in~\cite{soloperto2018learning}  and~\cite{rosolia2019sample}. Another class of corrective methods that ensure safety is based on control barrier functions (CBFs)~\cite{ames2019control,ames2016control,xu2015robustness}.
The use of CBFs with parametric uncertainty was recently considered in~\cite{lopez2020robust}. Similar to how the presented framework is planner agnostic, FaSTrack~\cite{herbert2017fastrack,fridovich2018planning}
is a reactive method for safe planning and control with the aim of fast real-time trajectory generation using reachability analysis.

Gaussian processes are a commonly used class of statistical models~\cite{williams2006gaussian}. These models are popular in learning-based control, in part, because they provide predictive distributions, which can characterize modeling errors. The authors in~\cite{berkenkamp2017safe} use the regularity of the uncertainty and the sufficient statistics of the learned GP models to safely expand the region of attraction and improve control performance. Safety is guaranteed by the existence of Lyapunov functions ensuring asymptotic stability. Similarly, the authors in~\cite{lederer2019uniform} propose a new method to compute uniform error bounds of learned GP models for safe control. Probabilistic chance constraint methods, which use uncertainty propagation, have been shown to provide both asymptotic and transient bounds on robot performance~\cite{koller2018learning}. The implementations that rely on approximate uncertainty propagation offer excellent empirical performance without theoretical guarantees, shown in~\cite{hewing2019cautious,ostafew2016robust}.

The proposed method avoids uncertainty propagation completely when considering nonlinear dynamics. Instead, we rely on uniform error bounds for GP predictions to apriori guarantee tracking performance with respect to a desired trajectory. Moreover, we provide an \textit{explicit} design for the feedback controller with stability and performance guarantees. This controller is capable of incorporating the learned dynamics while ensuring safety with respect to the updated models. This incorporation is based on both contraction theory~\cite{manchester2017control} and the $\mathcal{L}_1$ adaptive control theory~\cite{hovakimyan2010L1}. Safe planning and control using $\mathcal{L}_1$ adaptive control theory can be found in~\cite{pereida2018adaptive,pravitra2020l1,lakshmanan2020safe}. These results fall under the category of safe feedback motion planning, where the $\mathcal{L}_1$ adaptive controller is the ancillary controller that guarantees tracking performance and robustness for the planner's commanded tasks. In~\cite{pereida2018adaptive},  $\mathcal{L}_1$ is used to reject the system nonlinearity and the planner generates trajectories for the nominal linear model. In~\cite{pravitra2020l1}, a model predictive path integral (MPPI) planner considers nonlinear dynamics while the $\mathcal{L}_1$ controller provides robustness and compensates for the uncertainties. The effectiveness is shown empirically, without guarantees. The work in~\cite{lakshmanan2020safe} provides certifiable performance bounds for nonlinear nominal dynamics using contraction theory~\cite{manchester2017control}. Contraction theory enables tracking control with nonlinear systems without relying on feedback linearizability or strict-feedback structure, and thus can be used with learned dynamics as presented in this paper. Related work using contraction theory based control laws can be found in~\cite{manchester2017control,singh2017robust,lopez2020adaptive}.


\section{Preliminaries and Problem Statement}\label{sec:prelim}

We consider the following affine in control dynamics given by
\begin{subequations}\label{eqn:system_dynamics}
\begin{align}
    \dot{x}(t) =  F(\xi(t),x(t),u(t)) &=   \bar{F}(x(t),u(t)) + B(x(t))h(\xi(t),x(t)) \\
    &= f(x(t)) + B(x(t))(u(t) + h(\xi(t),x(t))), \quad x(0) = x_0,
\end{align}
\end{subequations}
where $x(t) \in \mathbb{R}^n$ is the system state and $u(t) \in \mathbb{R}^m$ is the control input. The functions $f(x(t)) \in \mathbb{R}^n$ and $B(x(t)) \in \mathbb{R}^{n \times m}$ are known. The function $h(\xi(t),x(t)) \in \mathbb{R}^m$ represents the model uncertainty within the range of the input operator $B(x(t))$. Here $\xi(t) \in \mathbb{R}^l$ represents any known time-varying parameter, including $\xi(t) = t$, allowing  to consider both time and state dependent uncertainties.
Finally, it is evident from~\eqref{eqn:system_dynamics} that $F(\xi(t),x(t),u(t))$ represents the \textit{actual dynamics}, whereas $\bar{F}(x(t),u(t))$ represents the \textit{nominal dynamics}.
\begin{assumption}\label{assmp:models:functions}
The functions $f(x)$, $B(x)$, and $h(\xi,x)$ are continuous, bounded, and Lipschitz in $x$, uniformly in $\xi$, for all $\xi \in \mathbb{R}^l$, and for all $x \in D \subset \mathbb{R}^n$, where $D$ is a compact set which can be arbitrarily large. Moreover, $B(x)$ has full column rank for all $x \in D$.
\end{assumption}
\begin{assumption}\label{assmp:models:functions_derivatives}
The derivatives $\pdv{f}{x}(x)$, $\pdv{B}{x}(x)$, $\pdv{h}{x}(\xi,x)$, and $\pdv{h}{\xi}(\xi,x)$ are  bounded for all $\xi \in \mathbb{R}^l$ for all $x \in D \subset \mathbb{R}^n$, where $D \in \mathbb{R}^n$ is any compact set allowed to be arbitrarily large.
\end{assumption}

For the planning problem to be feasible with respect to the robot's dynamic capabilities, we provide the following definition.
\begin{definition}\label{def:planner}
Over a \textit{planning horizon} $[0,T_f]$, $0 < T_f \leq \infty$, we say that $(x_d(t),u_d(t))$ is a \textit{desired state-input pair} if $\dot{x}_d(t) = \bar{F}(x_d(t),u_d(t))$ and $x_d(t) \in \mathcal{X}$, for all $t \in [0,T_f]$, where $\mathcal{X} \subset \mathbb{R}^n$ is any compact convex set.
Given any $\rho>0$, we define the $\rho$-norm ball centered at $x_d(t)$ by
\begin{equation} \label{eqn:norm_balls}
\Omega(\rho,x_d(t)) :=\{y \in \mathbb{R}^n~|~\norm{y - x_d(t)} \leq \rho\}.
\end{equation}
Here $\norm{\cdot}$ denotes the Euclidean norm. The norm balls induce the compact set
\begin{equation}\label{eqn:tube}
    \mathcal{O}_{x_d}(\rho) = \cup_{t \in [0,T_f]} \Omega (\rho,x_d(t)),
\end{equation} which we refer to as the tube.
\end{definition}
 Note that $\mathcal{X}$ is not a safe set since it does not consider the knowledge of the environment and obstacles, and it only represents the maximal limits of the state-space without which the planning cannot be feasible
 .
Additionally, since $\bar{F}$ (nominal dynamics) is known, any model-based planner can thus generate the desired pair $(x_d(t),u_d(t))$ satisfying the state-constraints.

Since we show that the actual state $x(t) \in \mathcal{O}_{x_d}(\rho)$ over the time-horizon of planning, to ensure feasibility, we place the following assumption.
\begin{assumption}\label{assmp:planner}
Given any tube width $\rho>0$ and planning horizon $[0,T_f]$, $0 < T_f \leq \infty$, the planner produces a state-input pair $(x_d(t),u_d(t))$ such that the induced tube $\mathcal{O}_{x_d}(\rho)$ defined in~\eqref{eqn:tube} satisfies
\[
\mathcal{O}_{x_d}(\rho) \in \mathcal{X}, \quad \forall t \in [0,T_f].
\] Furthermore, the desired control input $u_d(t)$ satisfies
\[
\norm{u_d(t)} \leq \Delta_{u_d}, \quad \forall t \in [0,T_f],
\] with the upper bound known.
\end{assumption}

Given any $\rho>0$, based on Assumptions~\ref{assmp:models:functions}-\ref{assmp:planner}, we have that for all $\xi \in \mathbb{R}^l$ and the compact convex set $\mathcal{X} \subset \mathbb{R}^n$
\begin{subequations}\label{eqn:model_bounds}
\begin{align}
    &\norm{f(x)} \leq \Delta_f,~\norm{\pdv{f}{x}(x)} \leq \Delta_{f_x},~\sum_{i=1}^n \norm{\pdv{B}{x_i}(x)} \leq \Delta_{B_x},~
\sum_{j=1}^m \norm{\pdv{[b]_{\cdot,j}}{x}(x)} \leq \Delta_{b_x}, \quad \forall x \in \mathcal{O}_{x_d}(\rho) \subset \mathcal{X},\\
&\norm{B^\dagger(x)} \leq \Delta_{B^\dagger},~\sum_{i = 1}^n \norm{ \pdv{B^\dagger}{x_i}(x)    } \leq \Delta_{B_x^\dagger}, \quad \forall x \in \mathcal{O}_{x_d}(\rho) \subset \mathcal{X},\\
&\norm{h(\xi,x)} \leq \Delta_h,~\norm{\nabla_x h(\xi,x)} \leq \Delta_{h_x},~\norm{\nabla_\xi h(\xi,x)} \leq \Delta_{h_\xi}, \quad \forall (x,\xi) \in \mathcal{X} \times \mathbb{R}^l,
\end{align}
\end{subequations} where where $[b]_{\cdot,j}(x)$ denotes the $j^{th}$ column of $B(x)$, $B^\dagger(x) = \left(B^\top(x) B(x)   \right)^{-1} B^\top (x)$ denotes the Moore-Penrose inverse which is guaranteed to exist by Assumption~\ref{assmp:models:functions}, and $\nabla$ denotes the gradient with respect to the sub-scripted variable.


\noindent \textbf{Problem Statement:}
Given the learned probabilistic estimates of the uncertainty $h(\xi,x)$, any desired state-input pair $(x_d(t),u_d(t))$, $t \in [0,T_f]$, designed by a planner using the nominal dynamics, and the desired robustness margins, the goal is to design the control input $u(t)$ that guarantees the existence of an apriori computable tube-width $\rho$  so that the state of the uncertain dynamics in~\eqref{eqn:system_dynamics} satisfies $x(t) \in \Omega(\rho,x_d(t)) \subset \mathcal{O}_{x_d}(\rho)$ with high probability, for all $t \geq 0$, from all initial conditions $x_0 \in D = \mathcal{X}$, while satisfying the robustness requirements. Importantly, the existence of the pre-computable tubes should not depend on the quality of the learned estimates, thus ensuring that \textit{safety remains decoupled from learning}. The learning should only affect the performance bounds and the optimality of the planned trajectory.

We now discuss the two constituent components of the $\mathcal{CL}_1$-$\mathcal{GP}$ control, namely Bayesian learning and the $\mathcal{CL}_1$ control.

\subsection{Bayesian learning} The probabilistic estimates of the uncertainty $h(\xi(t),x(t))$ in~\eqref{eqn:system_dynamics} are learned using GP regression. We place the following assumption to compute the prediction error bounds.
\begin{assumption}\label{assmp:GP}

We assume that each of the elements $[h]_i(\xi,x)$ are independent. Moreover, we assume that each element is a sample from a GP
\[
[h]_i(\xi,x) = [h]_i(z) \sim \mathcal{GP}(0,K_i(z,z')), \quad i \in \{1,\dots,m\},
\]
where $z = \begin{bmatrix}\xi^\top & x^\top \end{bmatrix}^\top \in \mathbb{R}^{l+n}$ and the kernel functions $K_i: \mathbb{R}^{(l+n) \times (l+n)} \rightarrow \mathbb{R}$ are known. Moreover, the kernels are twice-continuously differentiable with known constants $L_{K_i}$, $\nabla_\xi L_{K_i}$, $\nabla_x L_{K_i}$  such that
\[
L_{K_i} =
\max_{z,z' \in \mathcal{Z}}
\norm{
\nabla_z K_i(z,z')
}, \quad
\nabla_\xi L_{K_i} =  \max_{z,z' \in \mathcal{Z}} \norm{\nabla_\xi^2 K_i(z,z')}, \quad
\nabla_x L_{K_i} =  \max_{z,z' \in \mathcal{Z}} \norm{\nabla_x^2 K_i(z,z')},
\]
for $i \in \{1,\dots,m\}$, where $\mathcal{Z} = \mathcal{X}_\xi \times \mathcal{X}$, where $\mathcal{X}_\xi \subset \mathbb{R}^l$ is a convex compact set.
\end{assumption}
The assumption that the uncertainty is a sample from a GP with a known prior is less conservative than requiring the uncertainty to be a member of the reproducing kernel Hilbert space (RKHS) associated with the kernel. For example, sample functions of GPs with squared-exponential (SE) kernels correspond to continuous functions, whereas the associated RKHS space contains only analytic functions~\cite{van2011information}. Moreover, the constants assumed to exist in Assumption~\ref{assmp:GP} are easily computable, for example, for the often used squared-exponential (SE) kernel.

\begin{assumption}
$\xi(t) \in \mathcal{X}_\xi$
\end{assumption}

We assume that we have $N \in \mathbb{N}$ measurements of the form
\[
y_k = h(\xi_k,x_k) + \kappa  = h(z_k) + \kappa  = B^\dagger (x_k)\left(\dot{x}_k - f(x_k)  \right)-u_k + \kappa \in \mathbb{R}^m,
\quad k \in \{1,\dots,N\},
\]
where $\kappa$ represents measurement noise distributed normally as $\kappa \sim \mathcal{N}(0_m,\sigma^2 \mathbb{I}_m)$, where $0_m \in \mathbb{R}^m$ is a vector of zeros and $\mathbb{I}_m$ denotes the identity matrix of dimension $m$. Using the measurements, we set up the data as
\begin{equation}\label{eqn:data}
\mathbf{D} = \{\mathbf{Y},\mathbf{Z}\}, \quad \mathbf{Y} = \begin{bmatrix} y_1 & \cdots & y_N \end{bmatrix} \in \mathbb{R}^{m \times N}, \quad \mathbf{Z} = \begin{bmatrix} z_1 & \cdots & z_N \end{bmatrix} \in \mathbb{R}^{(l+n) \times N}.
\end{equation}
Thus, for each of the constituent functions $[h]_i$, $i \in \{1,\dots,m\}$, we have the data as $\mathbf{D}_i = \{[\mathbf{Y}]_{i,\cdot},\mathbf{Z}\}$, where $[\mathbf{Y}]_{i,\cdot}$ denotes the $i^{th}$-row of the matrix $\mathbf{Y}$. GP regression proceeds by conditioning the prior in Assumption~\ref{assmp:GP} on the measured data as in~\cite{williams2006gaussian} to obtain the posterior distribution at any test point $z^\star \in \mathcal{X}_\xi \times \mathcal{X}$ as
\begin{equation}\label{eqn:function:posterior}
\mathbb{R} \ni [h]_i(z^\star) \sim \mathcal{N}\left( \nu_{i,N}(z^\star),\sigma^2_{i,N}(z^\star) \right), \quad i \in \{1,\dots,m\},
\end{equation}
where the mean $\nu_{i,N}(z^\star)$ and the variance $\sigma^2_{i,N}(z^\star)$ are defined as
\begin{align*}
\nu_{i,N}(z^\star) = & K_i(z^\star,\mathbf{Z})^\top \left[K_i(\mathbf{Z},\mathbf{Z}) + \sigma^2 \mathbb{I}_N  \right]^{-1} \left( [\mathbf{Y}]_{i,\cdot}  \right)^\top,\\
\sigma^2_{i,N}(z^\star) = & K_i(z^\star,z^\star) -  K_i(z^\star,\mathbf{Z})^\top \left[K_i(\mathbf{Z},\mathbf{Z}) + \sigma^2 \mathbb{I}_N  \right]^{-1} K_i(z^\star,\mathbf{Z}).
\end{align*}
Furthermore, $ K_i(z^\star,\mathbf{Z}) \in \mathbb{R}^N$ and $K_i(\mathbf{Z},\mathbf{Z}) \in \mathbb{S}^N$ are defined as
\begin{align*}
    K_i(z^\star,\mathbf{Z}) =& \begin{bmatrix} K_i(z^\star,z_1) & \cdots &  K_i(z^\star,z_N)  \end{bmatrix}, \quad \left[K_i(\mathbf{Z},\mathbf{Z})\right]_{p,q} =  K_i(z_p,z_q), \quad (p,q) \in \{1,\dots,N\} \times \{1,\dots,N\}.
\end{align*}

Using the linearity of the differential operator, we can also compute the posterior distributions of the partial derivatives of $h(\xi,x)$ using the previously defined data in~\eqref{eqn:data}. The posterior distributions of the partial derivatives are given by
\begin{subequations}\label{eqn:function_derivative:posterior}
\begin{align}
\left(\nabla_\xi [h]_i(z^\star)\right)^\top   \sim & \mathcal{N}\left(\nabla_\xi \nu_{i,N}(z^\star)^\top, \nabla_\xi \sigma_{i,N}^2(z^\star)   \right),\\ \left(\nabla_x [h]_i(z^\star)\right)^\top  \sim & \mathcal{N}\left(\nabla_x \nu_{i,N}(z^\star)^\top, \nabla_x \sigma_{i,N}^2(z^\star)   \right),
\end{align}
\end{subequations}
where the mean functions $\nabla_\xi \nu_{i,N}(z^\star) \in \mathbb{R}^{1 \times l}$, $\nabla_x \nu_{i,N}(z^\star) \in \mathbb{R}^{1 \times n}$ and the variance functions $\nabla_\xi \sigma_{i,N}^2(z^\star)  \in \mathbb{S}^l$, $\nabla_x \sigma_{i,N}^2(z^\star) \in \mathbb{S}^n$ are defined as
\begin{align*}
    \nabla_\xi \nu_{i,N}(z^\star)^\top = & \left(\nabla_\xi K_i(z^\star,\mathbf{Z})\right)^\top \left[ K_i(\mathbf{Z},\mathbf{Z}) + \sigma^2 \mathbb{I}_N\right]^{-1}\left([\mathbf{Y}]_{i,\cdot}\right)^\top,\\
    \nabla_x \nu_{i,N}(z^\star)^\top =&  \left(\nabla_x K_i(z^\star,\mathbf{Z})\right)^\top \left[ K_i(\mathbf{Z},\mathbf{Z}) + \sigma^2 \mathbb{I}_N\right]^{-1}\left([\mathbf{Y}]_{i,\cdot}\right)^\top,                    \\
    \nabla_\xi \sigma_{i,N}^2(z^\star) =&  \nabla_{\xi,\xi'}^2 K_i (z^\star,z^\star) - \left(\nabla_\xi K_i(z^\star,\mathbf{Z})\right)^\top \left[ K_i(\mathbf{Z},\mathbf{Z}) + \sigma^2 \mathbb{I}_N\right]^{-1}\nabla_\xi K_i(z^\star,\mathbf{Z}),             \\
    \nabla_x \sigma_{i,N}^2(z^\star) =& \nabla_{x,x'}^2 K_i (z^\star,z^\star) - \left(\nabla_x K_i(z^\star,\mathbf{Z})\right)^\top \left[ K_i(\mathbf{Z},\mathbf{Z}) + \sigma^2 \mathbb{I}_N\right]^{-1}\nabla_x K_i(z^\star,\mathbf{Z}).
\end{align*}

Note that $\left(\nabla_\xi [h]_i(z^\star)\right)^\top \in \mathbb{R}^{l}$ and $\left(\nabla_x [h]_i(z^\star)\right)^\top \in \mathbb{R}^{n}$ are multivariate Gaussian random variables for each $i \in \{1,\dots,m\}$. Moreover, the individual elements of each are co-related in the general case as indicated by the presence of non-zero off-diagonal terms in the posterior covariance matrices $\nabla_\xi \sigma_{i,N}^2(z^\star)$ and $\nabla_x \sigma_{i,N}^2(z^\star)$. However, by marginalizing as in~\cite[Sec.~2.3.1]{bishop2006pattern} we can obtain the individual posterior distributions of each component as
\begin{subequations}\label{eqn:function_derivative:individual_posterior}
     \begin{align}
         \left[\nabla_\xi h\right]_{i,k}(z^\star)   \sim & \mathcal{N}\left(\left[\nabla_\xi \nu_{i,N}\right]_k(z^\star), \left[\nabla_\xi \sigma_{i,N}^2\right]_{k,k}(z^\star)   \right), \quad (i,k) \in \{1,\dots,m\} \times \{1,\dots,l\}, \\
         \left[\nabla_x h\right]_{i,k}(z^\star)  \sim & \mathcal{N}\left(\left[\nabla_x \nu_{i,N}\right]_k(z^\star), \left[\nabla_x \sigma_{i,N}^2\right]_{k,k}(z^\star)   \right), \quad (i,k) \in \{1,\dots,m\} \times \{1,\dots,n\}.
     \end{align}
\end{subequations}

\subsection{\rellone control}

In the presented methodology, the control input $u(t)$ is computed using the $\mathcal{CL}_1$ control as presented in~\cite{lakshmanan2020safe}. The $\mathcal{CL}_1$ control can be decomposed as
\begin{equation}\label{eqn:Rl1_input}
    u(t) = u_c(t) + u_a(t),
\end{equation}
where $u_c(t)$ is the control input designed for the known dynamics and relies on the contraction theoretic notion of Riemannian energy~\cite{lopez2020adaptive,singh2019robust,manchester2017control}, whereas $u_a(t)$ is the adaptive control input designed based on the $\mathcal{L}_1$ adaptive control theory~\cite{wang2012l1,hovakimyan2010L1} and is tasked with compensating for the model uncertainties. The existence of the $u_c(t)$ input relies on the existence of the \textit{control contraction metric (CCM)}~\cite{manchester2017control}, which is defined to be any smooth function $M(x)$, satisfying for all $(x,\delta_x) \in T \mathcal{X}$ (the tangent bundle of $\mathcal{X}$):
\begin{subequations}\label{eqn:CCM_conditions}
\begin{align}
   &\underline{\alpha} \mathbb{I}_n \succeq M(x) \succeq \bar{\alpha} \mathbb{I}_n, \\
   &\label{eqn:killing}\partial_{[b]_{\cdot,j}} M(x) + \left[M(x) \frac{\partial [b]_{\cdot,j}(x)}{\partial_x}  \right]_\mathbb{S} = 0, \quad j \in \{1,\dots,m\}, \\
   &\delta_x^\top M(x)B(x) = 0 \Rightarrow \delta_x^\top \left( \partial_f M(x) + \left[M(x) \frac{\partial f(x)}{\partial x} \right]_\mathbb{S} + 2 \lambda M(x)    \right)\delta_x \leq 0,
   \end{align}
\end{subequations}
for some scalars $\lambda > 0$, $0 < \underline{\alpha} < \bar{\alpha} < \infty$. Here $[b]_{\cdot,j}$ denotes the $j^{th}$ column of $B(x)$ and $\partial_f M(x)$ denotes the directional derivative of $M(x)$ with respect to $f(x)$. The same holds for $\partial_{[b]_{\cdot,j}} M(x)$. Moreover, $[A]_{\mathbb{S}}$ denotes the symmetric part of the matrix $A$.
Further details are presented in~\cite{manchester2017control} and~\cite{lakshmanan2020safe}. Note that the synthesis of the CCM $M(x)$ depends only on the nominal dynamics and can be computed offline. We place the following assumption.
\begin{assumption}\label{assmp:CCM}
The nominal dynamics $\bar{F}$ in~\eqref{eqn:system_dynamics} admit a CCM $M(x)$, for all $x \in \mathcal{X}$, and for some positive constants $\lambda$, $\underline{\alpha}$, and $\bar{\alpha}$ as in~\eqref{eqn:CCM_conditions}.
\end{assumption}
The CCM $M(x)$ defines a control Lypunov function (CLF) for the nominal system in the form of the Riemannian energy $\mathcal{E}(x_d(t),x(t))$~\cite{manchester2017control,lakshmanan2020safe}. Using the Riemannian energy and the given desired state-input pair $(x_d(t),u_d(t))$ computed using the nominal dynamics $\dot{x}_d(t) = \bar{F}(x_d(t),u_d(t))$, the input $u_c(t)$ is given by
\begin{equation}\label{eqn:CCM_input}
u_c(t)=u_d(t)+k_c(x_d(t),x(t)),
\end{equation}
where
\begin{subequations}\label{eqn:CCM_optimization}
\begin{align}
    &k_c(x_d(t),x(t))= \argmin_{k\in\mathbb{R}^m} \norm{k}^2,\\ &\text{s.t.}~2\bar{\gamma}_{\bar{s}}^{\top}(1,t)M(x(t))\dot{x}_k(t)-2\bar{\gamma}_{\bar{s}}^{\top}(0,t)M(x_d(t))\dot{x}_d(t)\leq -2\lambda\mathcal{E}(x_d(t),x(t)).
\end{align}
\end{subequations}
Here $\bar{\gamma}(\bar{s},t)$, $\bar{s} \in [0,1]$, is the minimal geodesic between $x_d(t)$ and $x(t)$ on the Riemannian manifold $(\mathcal{X},M)$ with $\bar{\gamma}(1,t) = x(t)$ and $\bar{\gamma}(0,t) = x_d(t)$. Furthermore, $\dot{x}_k(t) = \bar{F}(x(t),u_d(t) + k)$.
The quadratic program (QP) admits an analytical solution as explained in~\cite[Sec.~5.1]{singh2019robust}.

The $\mathcal{L}_1$ adaptive control input $u_a(t)$ in~\eqref{eqn:Rl1_input} relies on three components: the state-predictor, the adaptation law, and the control law. The \textit{state-predictor} is given by
\begin{equation}\label{eqn:state-predictor}
    \dot{\hat{x}}(t) = \bar{F}(x(t),u_c(t) + u_a(t) + \hat{\mu}(t)) + A_m \tilde{x}(t), \quad \hat{x}(0) = x_0,
\end{equation} where $\hat{x}(t)$ is the state of the predictor, $\tilde{x}(t) = \hat{x}(t) - x(t)$ is the state prediction error, and $A_m \in \mathbb{R}^{n \times n}$ is an arbitrary Hurwitz matrix. The \textit{uncertainty estimate} $\hat{\mu}(t)$ is driven by the state prediction error via the following \textit{adaptation law}
\begin{equation}\label{eqn:adaptation_law}
\dot{\hat{\mu}}(t) = \Gamma \text{Proj}_\mathcal{H} \left(\hat{\mu}(t), -B^\top (x) P \tilde{x}(t)   \right), \quad \hat{\mu}(0) \in \mathcal{H},
\end{equation}
where $\Gamma > 0$ is the adaptation rate, $\mathcal{H} = \{y \in \mathbb{R}^m~|~\norm{y} \leq \Delta_h\}$ is the conservative set within which the uncertainty estimate is restricted to lie in with $\Delta_h$ defined in~\eqref{eqn:model_bounds}. Additionally, $\mathbb{S}^n \ni P \succ 0$ is the solution to the Lyapunov equation $A_m^\top P + P A_m = -Q$, for some $\mathbb{S}^n \ni Q \succ 0$. Finally, $\text{Proj}_\mathcal{H}(\cdot,\cdot)$ is the standard projection operator~\cite{hovakimyan2010L1}. Finally, the input $u_a(t)$ is defined via the following \textit{control law} presented using the Laplace transform
\begin{equation}\label{eqn:control_law}
u_a(s) = -C(s)\hat{\mu}(s),
\end{equation} where $C(s)$ is a low-pass filter with bandwidth $\omega$ and satisfies $C(0) = \mathbb{I}_m$. Note that we use the variable $s$ to represent both the Laplace variable and the geodesic parameter in~\eqref{eqn:CCM_optimization}. The distinction is clear from context.

We now briefly explain how the $\mathcal{CL}_1$ guarantees safety by the existence of pre-computable tubes using only the available conservative knowledge presented in Assumptions~\ref{assmp:models:functions} and~\ref{assmp:models:functions_derivatives}, i.e., without any learning. For arbitrarily chosen positive scalars $\rho_a$ and $\epsilon$, and define
\begin{equation}\label{eqn:tube_def}
 \rho_r = \sqrt{\frac{\bar{\alpha}}{\underline{\alpha}}} \norm{x_d(0) - x_0} + \epsilon, \quad \rho = \rho_r + \rho_a,
\end{equation} where $\bar{\alpha}$ and $\underline{\alpha}$ are defined in Assumption~\ref{assmp:CCM}. Let us define the following constants
\begin{subequations}\label{eqn:RL1_bounds}
\begin{align}
\label{eqn:bounds:M_x}
\Delta_{M_x} &:= \sup_{x \in \mathcal{O}_{x_d}(\rho)} \sum_{i=1}^n \norm{\pdv{M}{x_i}(x)}, \\
\label{eqn:bounds:psi_x}
\Delta_{\Psi_x} &:= 2\Delta_{B_x} + \frac{\Delta_B \Delta_{M_x}}{\ualpha},\\
\label{eqn:bounds:delta_u}
\Delta_{\delta_u} &:= \frac{1}{2}\sup_{x \in \mathcal{O}_{x_d}(\rho)}\left(\frac{\olambda(L^{-\top}(x)Z(x)L^{-1}(x))}{\usigma_{> 0}(B^\top(x)L^{-1}(x))}\right),\\
 \label{eqn:bounds:dx_r}
\Delta_{\dot{x}_r} &:= \Delta_f + \Delta_B(\norm{\mathbb{I}_m - C(s)}_{\mathcal{L}_1}\Delta_h + \Delta_{u_d} + \rho\Delta_{\delta_u}), \\
\label{eqn:bounds:dx}
\Delta_{\dot{x}} &:= \Delta_f + \Delta_B(2\Delta_h + \Delta_{u_d} + \rho\Delta_{\delta_u}),\\
\label{eqn:bounds:tx}
\Delta_{\tilde{x}} &:= \sqrt{\frac{4\olambda(P)\Delta_h(\Delta_{h_\xi} + \Delta_{h_x}\Delta_{\dot{x}})}{\ulambda(P)\underline{\lambda}(Q)} + \frac{4\Delta_h^2}{\underline{\lambda}(P)}}, \\
\label{eqn:bounds:teta}
\Delta_{\tilde{\eta}} &:= \left(\Delta_{B^\dagger_x} \Delta_{\dot{x}} + (\norm{sC(s)}_{\mathcal{L}_1} + \norm{A_m}) \Delta_{B^\dagger} \right)\Delta_{\tilde{x}},\\
\label{eqn:bounds:theta}
\Delta_\theta &:= \frac{\Delta_B \overline{\alpha}  \Delta_{\tilde{\eta}}}{\lambda},\\
\label{eqn:bounds:dpsi}
\Delta_{\dot{\Psi}} &:= \oalpha\left(
    \Delta_B \Delta_{\dot{\ogamma}_{\bar{s}}} + \frac{\Delta_B \Delta_{M_x} \Delta_{\dot{x}}}{\sqrt{\oalpha\ualpha}} + \Delta_{B_x} \Delta_{\dot{x}}\right), \\
\label{eqn:bounds:dgamma_s}
\Delta_{\dot{\ogamma}_{\bar{s}}} &:= \sqrt{\frac{\oalpha}{\ualpha}}\left(\Delta_{f_x} + (\Delta_h + \Delta_{u_d
} + \rho \Delta_{\delta_u})\Delta_{b_x} + \left(\Delta_{h_x} + \frac{\sqrt{\ualpha}\Delta_{\delta_u}}{\sqrt{\oalpha}}\right) \Delta_B \right),
\end{align}
\end{subequations}
where $\mathcal{O}_{x_d}(\rho)$ is defined in~\eqref{eqn:tube}; $\Delta_{u_d}$, $\Delta_f$, $\Delta_{f_x}$, $\Delta_B$, $\Delta_{B_x}$, $\Delta_{b_x}$, $\Delta_h$, $\Delta_{h_\xi}$, $\Delta_{h_x}$, $\Delta_{B^\dagger}$ and $\Delta_{B_x^\dagger}$ are defined in~\eqref{eqn:model_bounds}; $\oalpha$ and $\ualpha$ are defined in \cref{assmp:CCM}; and $Z(x)$ is defined as
\[
Z(x):=-\partial_{f} W(x)+ 2\Sym{\pdv{f}{x}(x)W(x)} + 2 \lambda W(x),
\]
where $W(x) = M(x)^{-1}$ is referred to as the dual metric and $L(x)^\top L(x) = W(x)$, and these entities are guaranteed to exist due to the positive definiteness of the CCM $M(x)$. As before, $\partial_{f} W(x)$ denotes the directional derivative of the dual metric $W(x)$ with respect to $f(x)$~\cite{manchester2017control}. Furthermore, $\norm{sC(s)}_{\mathcal{L}_1}$ denotes the $\mathcal{L}_1$ function norm of the impulse response of $sC(s)$~\cite[Sec.~A.7]{hovakimyan2010L1}. Finally, for any real-valued matrices $A$ and $B$, with $B$ square, $\usigma_{> 0} (A)$, $\olambda (B)$, and $\ulambda (B)$, denote the smallest singular-value of $A$, and the largest and smallest eigenvalues of $B$, respectively.
For the purposes of analysis, we need the following constants
\begin{subequations}\label{eqn:RL1:kappas}
\begin{align}
    \kappa_1(\Delta_h,\Delta_{h_x},\Delta_{h_\xi}) =& 2 \rho \Delta_B \frac{\bar{\alpha}}{\underline{\alpha}} \left(
    \frac{\Delta_h}{|2\lambda/\omega - 1|} + \frac{\Delta_{h_\xi} + \Delta_{h_x} \Delta_{\dot{x}_r}}{2 \lambda},
    \right), \\
    \kappa_2(\Delta_h,\Delta_{h_x},\Delta_{h_\xi}) =& \bar{\alpha} \Delta_{\Psi_x} \frac{\bar{\alpha}}{\underline{\alpha}} \left(
    \frac{\Delta_h}{|2\lambda/\omega - 1|} + \frac{\Delta_{h_\xi} + \Delta_{h_x} \Delta_{\dot{x}_r}}{2 \lambda}
    \right),\\
    \kappa_3(\Delta_h,\Delta_{h_x}) =& \bar{\alpha} \Delta_{h_x} \left(
    \frac{4 \lambda \Delta_B + \Delta_{\dot{\Psi}}}{\lambda}
    \right),\\
    \kappa_4(\Delta_h,\Delta_{h_x},\Delta_{h_\xi}) =& \Delta_\theta,
\end{align}
\end{subequations} using which we further define
\begin{align}\label{eqn:RL1_constants}
    \zeta_1(\kappa_1,\omega) =& \frac{\kappa_1(\Delta_h,\Delta_{h_x},\Delta_{h_\xi})}{\omega},~ \zeta_2(\kappa_2,\omega) = \frac{\kappa_2(\Delta_h,\Delta_{h_x},\Delta_{h_\xi})}{\omega},~ \zeta_3(\kappa_3,\omega) = \frac{\kappa_3(\Delta_h,\Delta_{h_x})}{\omega}.
\end{align}
 Note that $\kappa_1-\kappa_3$ are monotonically increasing as a function of the conservative known bounds $\Delta_h$, $\Delta_{h_x}$, and $\Delta_{h_\xi}$, and vanish for zero uncertainty bounds.

 The rate of adaptation $\Gamma$ in~\eqref{eqn:adaptation_law} and filter bandwidth $\omega$ in~\eqref{eqn:control_law} then need to satisfy
\begin{equation}\label{eqn:RL1_conditions}
    \rho_r^2 \geq \frac{\mathcal{E}(x_d(0),x_0)}{\underline{\alpha}} + \zeta_1(\kappa_1,\omega), ~\underline{\alpha} > \zeta_2(\kappa_2,\omega) + \zeta_3(\kappa_3,\omega), ~ \sqrt{\Gamma} > \frac{\kappa_4(\Delta_h,\Delta_{h_x},\Delta_{h_\xi})}{\rho_a ( \underline{\alpha} - \zeta_2(\kappa_2,\omega) - \zeta_3(\kappa_3,\omega)   )}.
\end{equation}
Since the constants $\zeta_1(\omega),\zeta_2(\omega),\zeta_3(\omega) \propto 1/\omega$, the conditions in~\eqref{eqn:RL1_conditions} can always be satisfied by choosing an appropriately large adaptation rate $\Gamma$ and filter bandwidth $\omega$. The following result quantifies the existence of safe tubes around any desired state $x_d(t)$. The detailed proof of the following theorem can be found in~\cite[Thm.~5.1]{lakshmanan2020safe}.
\begin{theorem}[\cite{lakshmanan2020safe}]\label{thm:RL1}
Let Assumptions~\ref{assmp:models:functions},~\ref{assmp:models:functions_derivatives}, and~\ref{assmp:CCM} hold and let the filter bandwidth $\omega$ and rate of adaptation $\Gamma$ satisfy~\eqref{eqn:RL1_conditions}. Given any desired state input pair $(x_d(t),u_d(t))$ satisfying the nominal dynamics $\dot{x}_d(t) = \bar{F}(x_d(t),u_d(t))$, the state of the actual uncertain system~\eqref{eqn:system_dynamics}, $\dot{x}(t) = F(\xi(t),x(t),u(t))$, with control input~\eqref{eqn:Rl1_input} satisfies
\[
x(t) \in \Omega(\rho,x_d(t)) \subset \mathcal{O}_{x_d}(\rho), \quad \forall t \geq 0.
\]
Furthermore, the state $x(t)$ is uniformly ultimately bounded as
\[
x(t) \in \Omega(\delta(\omega,T),x_d(t)) \subset \Omega(\rho,x_d(t)), \quad \forall t \geq T > 0,
\]
where the uniform ultimate bound is given by
\[
\delta(\omega,T) = \mu(\omega,T) + \rho_a, \quad \mu(\omega,T) = \sqrt{e^{-2\lambda T}\mathcal{E}(x_d(0),x_0)/\underline{\alpha}  + \zeta_1(\kappa_1,\omega)            }.
\]
\end{theorem}


\section{\rellone-\gp : Riemannian Energy \ellone with Gaussian Process Learning}
\definecolor{lcolor}{rgb}{0.79,0.88,1.0}
\definecolor{bcolor}{RGB}{252,184,203}
\definecolor{ccolor}{RGB}{234,208,255}
\tikzstyle{sqblock} = [draw, fill=bcolor!20, rectangle,
    minimum height=3em, minimum width=3em]
\tikzstyle{block} = [draw, fill=bcolor!20, rectangle,
    minimum height=3em, minimum width=6em]
\tikzstyle{sum} = [draw, fill=bcolor!20, circle, node distance=1cm]
\tikzstyle{input} = [coordinate]
\tikzstyle{output} = [coordinate]
\tikzstyle{phantom} = [coordinate]

\begin{figure}[t]
  \centering
  \captionsetup{justification=centering}
    \begin{tikzpicture}[scale=1, transform shape,align=center] 

    \filldraw[fill=lcolor!70, densely dotted] (3.4,0.8) -- (7,0.8) -- (7,-4.5) -- (12.5,-4.5) -- (12.5,-7.5) -- (3.4,-7.5) -- cycle;
    \filldraw[fill=lcolor!30,draw=ccolor!30] (3.4, -7.5) rectangle (12.5,-8.5) node[pos=.5] {$\mathcal{CL}_1$ controller};

       \node [input, name=input] {};
        \node [block] (planner) [right=0.7cm of input] {Planner};
        \node [block] (ccm) [right=1.2cm of planner] {Riemannian\\Feedback};
        \node [sum] (sum) [below=1cm of ccm] {};
        \node [block] (sys) [right=2.6cm of sum] {Uncertain System};
        \node [sqblock] (filt) [below=2.8cm of sum] {$C(s)$};
        \node [block] (pred) [right=2.5cm of filt] {State Predictor};
        \node [block] (learn) [below=0.8cm of sys] {Bayesian Learner};
        \node [block] (adap) [below=0.5cm of pred] {Adaptation Law};
        \node [sum] (diff) [right=1.0cm of pred] {};
        \node [sum] (sum2) [right=0.75cm of filt] {};
        \node [phantom] (p2) at (sys -| diff) {};
        \node [phantom, above=1.0cm of p2] (p3) {};
        \node [phantom, below=1.9cm of p2] (p4) {};
        \node [phantom, above=2.8cm of p2] (p5) {};
        \node [phantom, right=1.3cm of sum] (p6) {};
        \node [phantom, below=1.9cm of p6] (p7) {};
        \node [phantom, above=2.5cm of sys] (p8) {};
        \node [phantom] (p9) at (adap -| sum2) {};
        \node [phantom] (p10) at (sys |- ccm) {};

        \draw[-Latex] (planner) -- node [midway,above] {$x_d, u_d$} (ccm);
        \draw[-Latex] (ccm) -- node [midway,left] {$u_{c,\hat F}$} (sum);
        \draw[-Latex] (sum) -- node [midway,above] {$u$} (sys);
        \draw[-Latex] (filt) -- node [midway,left] {$u_{a,\hat F}$} node[pos=0.99,left] {$+$} (sum);
        \draw[-Latex] (p9) -|  node[pos=0.99,left] {$+$} (sum2);
        \draw[-Latex] (p4) -- (learn);
        \draw[-Latex] (pred) -- node [midway,above] {$\hat{x}$} (diff);
        \draw[-Latex] (diff) |- node [near start,right] {$\tilde{x}$} (adap);
        \draw[-Latex] (p2) -- node[pos=0.99,right] {$-$} (diff);
        \draw (p2) -- (p5);
        \draw[-Latex] (p5) -| (ccm);
        \draw (p6) -- (p7);
        \draw[-Latex] (p7) -- (learn);
        \draw (sys) -- node [midway,above] {$x$} (p2);
        \node [output, right=1.0cm of p2] (output) {};
        \draw[-Latex] (p2) -- (output);
        \draw[dash dot,{Circle[black,length=1.5pt]}-Latex] (learn) -- (pred);
        \draw[dash dot,{Circle[black,length=1.5pt]}-] (learn) -- (sys);
        \draw[dash dot,-Latex] (p8) -| (planner);
        \draw[dash dot] (p8) -- (p10);
        \draw[dash dot,-Latex] (p10) -- (ccm);
        \draw[dash dot] (sys) |- node [near start,right]
        {$\nu_N(\xi ,x)$} (p10);
        \draw[-Latex] (sum2) -- (pred);
        \draw[-Latex] (p7) -- (sum2);
        \draw (adap) -| node [near start,above] {$\hat{\mu}$}(p9);
        \draw[-Latex] (p9) -| (filt);

    \end{tikzpicture}
  \caption{Architecture of \rellone-\gp control}
  \label{fig:RL1GP}
\end{figure}
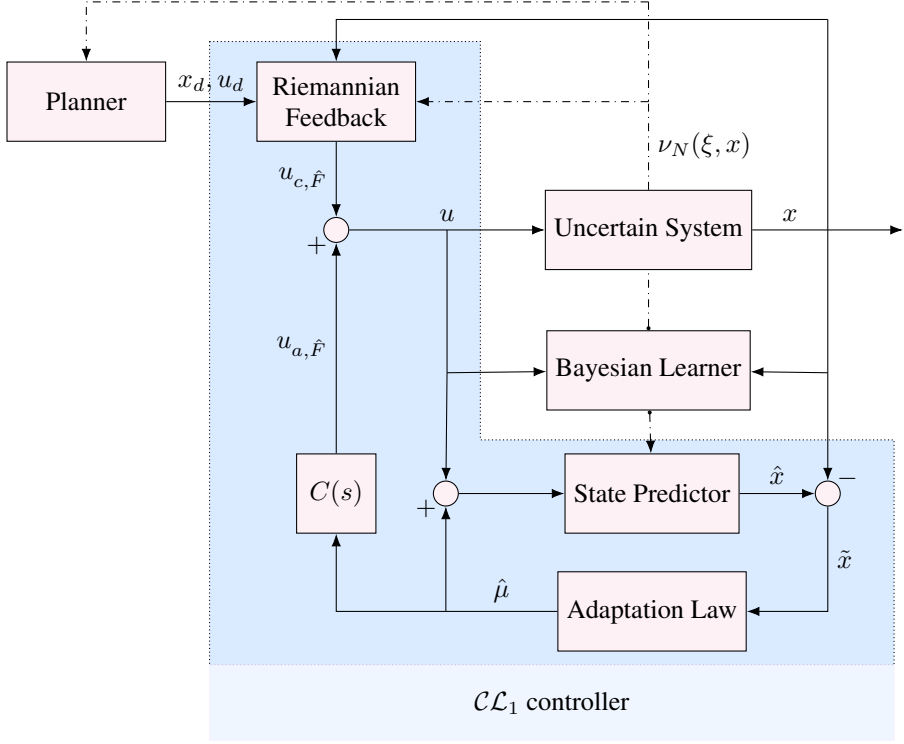

We now present the \rellone-\gp control framework illustrated in ~\Cref{fig:RL1GP}. As presented in Theorem~\ref{thm:RL1}, the $\mathcal{CL}_1$ control by itself can ensure the safety of the nominal system $\bar{F}$ by using the conservative knowledge of the uncertainty presented in Assumptions~\ref{assmp:models:functions} and~\ref{assmp:models:functions_derivatives}. Moreover, performance can be gained by increasing the bandwidth of the low-pass filter. However, this comes at the expense of the robustness, and optimality is not improved. In this section, we show how Bayesian learning in the form of GP posterior distributions can be incorporated withing the $\mathcal{CL}_1$ framework. The goals of incorporating learning are two-fold, i) providing the planner with higher fidelity models so as to produce trajectories with improved optimality, and ii) improving the tracking performance without requiring the tradeoff with robustness.

Given the posterior distribution of $h(z) = h(\xi,x)$ in~\eqref{eqn:function:posterior}, we may re-write the uncertain dynamics in~\eqref{eqn:system_dynamics} as
\begin{subequations}\label{eqn:learned_dynamics}
\begin{align}
    &\dot{x}(t) = F(\xi(t),x(t),u(t)) =  f(x(t)) + B(x(t))\left(u(t) + h(\xi(t),x(t))  \right) \\
    &=  f(x(t)) + B(x(t))\nu_N(\xi(t),x(t)) +  B(x(t))\left(u(t) + h(\xi(t),x(t)) - \nu_N(\xi(t),x(t)) \right)\\
    &=  \hat{F}(\xi(t),x(t),u(t)) + B(x(t))\left(u(t) + h(\xi(t),x(t)) - \nu_N(\xi(t),x(t)) \right),
\end{align}
\end{subequations} where $\hat{F}(\xi,x,u) = f(x) + B(x)(u + \nu_N(\xi,x))$ represents the \textit{learned dynamics}, which we obtain by adding and subtracting $\nu_N(\xi,x)$ within the control channel. Here, we denote
\begin{equation}\label{eqn:concatenated_mean}
    \nu_N(\xi,x) = \nu_N(z) = \begin{bmatrix} \nu_{1,N}(z) & \cdots & \nu_{m,N}(z)  \end{bmatrix}^\top \in \mathbb{R}^m.
\end{equation}
As in~\eqref{eqn:system_dynamics}, $F(\xi,x,u)$ represents the actual dynamics, but now $\hat{F}(\xi,x,u)$ represents the learned dynamics as opposed to $\bar{F}(\xi,x)$ representing the conservative nominal model.

Consider a desired state-input pair $(x_d(t),u_d(t))$, which is now designed by the planner using the learned dynamics $\hat{F}$, i.e., $\dot{x}_d(t) = \hat{F}(\xi(t),x_d(t),u_d(t))$, as opposed to the nominal dynamics $\bar{F}$. Note that $\hat{F}$ contains the nominal dynamics and the mean dynamics of the GP predictive distribution in~\eqref{eqn:function:posterior}. That is, $\hat{F}$ is deterministic, and therefore any planner that is being used does not have to rely on uncertainty propagation to ensure safety.
The goal now is to design the input $u(t)$ such that the state $x(t)$ of~\eqref{eqn:learned_dynamics} tracks $x_d(t)$ while remaining inside of pre-computable tube with high probability.
Similar to the $\mathcal{CL}_1$ input in~\eqref{eqn:Rl1_input}, the \rellone-\gp input is composed as
\[
u(t) = \hat{u}_{c,\hat{F}}(t) + \hat{u}_{a,\hat{F}}(t),
\]
where the individual components mirror those in~\eqref{eqn:CCM_input},~\eqref{eqn:state-predictor},~\eqref{eqn:adaptation_law}, and~\eqref{eqn:control_law}, but are now designed for the learning-based representation of the dynamics in~\eqref{eqn:learned_dynamics}. However, the major distinction is that $\hat{u}_{c,\hat{F}}(t)$ is designed to track $x_d(t)$ using the learned dynamics $\hat{F}$ (as opposed to the nominal dynamics $\bar{F}$), and the adaptive input $\hat{u}_{a,\hat{F}}(t)$ now compensates for the \textit{remainder uncertainty} $h - \nu_N$  as opposed to the uncertainty $h$. Therefore, we need to quantify the `size' of the remainder uncertainty $h - \nu_N$ for controller design. For this purpose, we will use the posterior distributions of the uncertainty and its derivatives in~\eqref{eqn:function:posterior} and~\eqref{eqn:function_derivative:posterior}, respectively, to compute high probability estimates of these bounds. In particular, we use the recent results in~\cite{lederer2019uniform}. We begin by presenting the following definition.
\begin{definition}\label{def:uniform_bounds}
Given any $\tau>0$, we define by $\mathcal{Z}_\tau \subset \mathcal{Z} = \mathcal{X}_\xi \times \mathcal{X}$, any discrete set satisfying
\[
|\mathcal{Z}_\tau| < \infty, \quad \text{and} \quad \max_{z \in \mathcal{Z}} \min_{z' \in \mathcal{Z}_\tau} \norm{z - z'} \leq \tau,
\] where $|\mathcal{Z}_\tau|$ denotes the cardinality of the set $\mathcal{Z}_\tau$. Given the posterior distribution of the uncertainty estimate in~\eqref{eqn:function:posterior}, $\nu_N(z) \in \mathbb{R}^m$ is defined as in~\eqref{eqn:concatenated_mean}, and
\[
\sigma_N(z) = \begin{bmatrix}
\sigma_{1,N}(z) & \cdots & \sigma_{m,N}(z)
\end{bmatrix}^\top \in \mathbb{R}^m.
\]
Similarly, for the posterior distribution of the partial derivatives of the uncertainty in~\eqref{eqn:function_derivative:posterior} and the marginal distributions in~\eqref{eqn:function_derivative:individual_posterior}, we define $\Sigma_{\xi}^{i,N}(z) \in \mathbb{R}^l$ and $\Sigma_{x}^{i,N}(z) \in \mathbb{R}^n$ as
\begin{align}
    \left[ \Sigma_{\xi}^{i,N}(z) \right]_p = & \left[\nabla_\xi \sigma_{i,N}(z) \right]_{p,p}  = \sqrt{\left[\nabla_\xi \sigma^2_{i,N}(z) \right]_{p,p}} , \quad
    \left[ \Sigma_{x}^{i,N}(z) \right]_q =  \left[\nabla_x \sigma_{i,N}(z) \right]_{q,q} = \sqrt{\left[\nabla_x \sigma^2_{i,N}(z) \right]_{q,q}},
\end{align} for  $p \in \{1,\dots,l\}$, and $q \in \{1,\dots,n\}$.

Finally, for any $\delta \in (0,1)$,  let $M(\tau,\mathcal{Z})$ denote the $\tau$-covering number of $\mathcal{Z}$\footnote{The $\tau$-covering number of $\mathcal{Z}$ is defined as the minimum number of $\tau$-norm balls that completely cover the set $\mathcal{Z}$. As explained in~\cite{lederer2019uniform}, upper bounds for the covering number can be computed trivially. For example, for a hypercube set $\mathcal{Z} \subset \mathbb{R}^n$, $M(\tau,\mathcal{Z}) \leq \left( 1 + r/\tau   \right)^{n}$, where $r$ is the edge length of the hypercube}, using which we define
\[
\beta(\tau) = 2 \log \left( \frac{mM(\tau,\mathcal{Z})}{\delta} \right), \quad
\beta_\xi(\tau) = 2 \log \left( \frac{lmM(\tau,\mathcal{Z})}{\hat{\delta}} \right), \quad
\beta_x(\tau) = 2 \log \left( \frac{nmM(\tau,\mathcal{Z})}{\hat{\delta}} \right),
\] and where $\hat{\delta} = 1 - (1-\delta)^{\frac{1}{m}}$.
\end{definition}

We now present the following theorem in which we compute high probability bounds for the remainder uncertainty. The following result is a generalization of~\cite[Thm.~3.1]{lederer2019uniform}. The proof is presented in Appendix~\ref{app:uniform_bounds}.
\begin{theorem}\label{thm:uniform_bounds}
Let Assumption~\ref{assmp:GP} hold and consider the posterior distributions of the uncertainty in~\eqref{eqn:function:posterior}~-~\eqref{eqn:function_derivative:individual_posterior} and any $\delta \in (0,1)$ and $\tau >0$.  Let $z = \begin{bmatrix} \xi^\top & x^\top  \end{bmatrix}^\top \in \mathcal{Z} = \mathcal{X}_\xi \times \mathcal{X}$ and define
\begin{align*}
    \Delta h (z,\tau) = & \sqrt{\beta(\tau)} \norm{\sigma_N(z)} + \gamma(\tau),\quad
    \nabla_\xi \Delta h(z,\tau) =   \sqrt{\sum_{i=1}^m \left(\nabla_\xi \gamma_i (\tau) + \sqrt{\beta_\xi (\tau)} \norm{\Sigma_{\xi}^{i,N}(z) }       \right)^2 },\\
    \nabla_x \Delta h(z,\tau) = & \sqrt{\sum_{i=1}^m \left(\nabla_x \gamma_i (\tau) + \sqrt{\beta_x (\tau)} \norm{\Sigma_{x}^{i,N}(z) }       \right)^2 },
\end{align*} where $\beta(\tau)$, $\beta_\xi(\tau)$, $\beta_x(\tau)$, $\sigma_N(z)$, and $\Sigma_{\xi}^{i,N}(z)$, $\Sigma_{x}^{i,N}(z)$, $i \in \{1,\dots,m\}$, are presented in Definition~\ref{def:uniform_bounds}.
Furthermore,
\begin{align*}
\gamma(\tau) =& \left(\Delta_{h_x} + \Delta_{h_\xi} + L_{\nu_N}  \right)\tau + \sqrt{\beta(\tau)} \omega_N \left(\tau \right),
\quad \nabla_\xi \gamma_i (\tau) = \left(\nabla_\xi \Delta_{h_\xi}^i + \nabla_\xi L_{i,\nu_N}  \right)\tau + \sqrt{\beta_\xi (\tau)} \nabla_\xi \omega_{i,N}(\tau),\\
    \nabla_x \gamma_i (\tau) =& \left(\nabla_x \Delta_{h_x}^i + \nabla_x L_{i,\nu_N}  \right)\tau + \sqrt{\beta_x (\tau)} \nabla_x \omega_{i,N}(\tau),
\end{align*}
where
\begin{align*}
    L_{\nu_N} =& \sqrt{ N \left( \sum_{i=1}^m  L_{K_i}^2
\norm{ \left[K_i(\mathbf{Z},\mathbf{Z}) + \sigma^2 \mathbb{I}_N  \right]^{-1} \left([\mathbf{Y}]_{i,\cdot}\right)^\top }^2 \right)}, \\
\omega_N\left(\norm{z - z'}\right)
=& \sqrt{ 2 \norm{z - z'} \sum_{i=1}^m   L_{K_i}  \left(1 + N \norm{\left[K_i(\mathbf{Z},\mathbf{Z}) + \sigma^2 \mathbb{I}_N  \right]^{-1}} \max_{z,z' \in \mathcal{Z}} K_i(z,z')  \right) },
\end{align*}
and for $i \in \{1,\dots,m\}$,
\begin{align*}
        \nabla_\xi L_{i,\nu_N} =& \sqrt{N} \nabla_\xi L_{K_i} \norm{\left[ K_i(\mathbf{Z},\mathbf{Z}) + \sigma^2 \mathbb{I}_N\right]^{-1}\left([\mathbf{Y}]_{i,\cdot}\right)^\top}, \\
    \nabla_x L_{i,\nu_N} =& \sqrt{N} \nabla_x L_{K_i} \norm{\left[ K_i(\mathbf{Z},\mathbf{Z}) + \sigma^2 \mathbb{I}_N\right]^{-1}\left([\mathbf{Y}]_{i,\cdot}\right)^\top},
\end{align*}
and
\begin{align*}
    \nabla_\xi \omega_{i,N} \left( \norm{z - z'} \right) =&   \sqrt{ 2 \norm{z - z'} \nabla_\xi L_{K_i} \left( 1 + N \norm{\left[ K_i(\mathbf{Z},\mathbf{Z}) + \sigma^2 \mathbb{I}_N\right]^{-1}}   \sum_{k=1}^l   \max_{z,z' \in \mathcal{Z}} \left| \pdv{K_i}{\xi_k}(z,z')\right| \right)  },\\
 \nabla_x \omega_{i,N} \left( \norm{z - z'} \right)
 =&   \sqrt{ 2 \norm{z - z'} \nabla_x L_{K_i} \left( 1 + N \norm{\left[ K_i(\mathbf{Z},\mathbf{Z}) + \sigma^2 \mathbb{I}_N\right]^{-1}}  \sum_{k=1}^n   \max_{z,z' \in \mathcal{Z}} \left| \pdv{K_i}{x_k}(z,z')\right| \right)  }.
\end{align*} Here, $N$ is the size of the data set $\mathbf{D} = \{\mathbf{Y},\mathbf{Z}\}$ in~\eqref{eqn:data}, and the constants $\Delta_h$, $\Delta_{h_x}$, $\Delta_{h_\xi}$, $\nabla_\xi \Delta_{h_\xi}^i$, and $\nabla_x \Delta_{h_x}^i$ are defined in~\eqref{eqn:model_bounds}, and $L_{K_i}$, $\nabla_x L_{K_i}$, and $\nabla_\xi L_{K_i}$ are defined in Assumption~\ref{assmp:GP}.

Then
\begin{align*}
    \Pr \left\{ \norm{h(z) - \nu_N(z)} \leq \Delta h (z,\tau), \quad \forall z \in \mathcal{Z}   \right\} \geq 1 - \delta,&\\
    \Pr \left\{ \norm{\nabla_x h(z) - \nabla_x \nu_N(z)} \leq \nabla_x \Delta h (z,\tau), \quad \forall z \in \mathcal{Z}   \right\} \geq 1 - \delta,&\\
    \Pr \left\{ \norm{\nabla_\xi h(z) - \nabla_\xi \nu_N(z)} \leq \nabla_\xi \Delta h (z,\tau), \quad \forall z \in \mathcal{Z}   \right\} \geq 1 - \delta,&
\end{align*} where $\nu_N(z) = \nu_N(\xi,x)$ is presented in Definition~\ref{def:uniform_bounds}.
\end{theorem}


Using the high-probability uniform bounds in Theorem~\ref{thm:uniform_bounds}, let us define
\begin{subequations}\label{eqn:GP_bounds_new}
\begin{align}
    \Delta_{\hat{h}} &= \sup_{z \in \mathcal{X}_\xi \times \mathcal{X}} \Delta h (z,\tau), \\
    \Delta_{\hat{h}_x} &= \sup_{z \in \mathcal{X}_\xi \times \mathcal{X}} \nabla_x \Delta h (z,\tau), \\
    \Delta_{\hat{h}_\xi} &= \sup_{z \in \mathcal{X}_\xi \times \mathcal{X}} \nabla_\xi \Delta h (z,\tau).
\end{align}
\end{subequations} Then, using Theorem~\ref{thm:uniform_bounds}, we conclude that, with probability at least $1 - \delta$
\begin{equation}\label{eqn:GP_bounds}
\Delta_{\hat{h}} \geq \norm{h(\xi,x) - \nu_N(\xi,x)},~ \Delta_{\hat{h}_x} \geq \norm{\nabla_x (h(\xi,x) - \nu_N(\xi,x))}, ~\Delta_{\hat{h}_\xi} \geq   \norm{\nabla_\xi (h(\xi,x) - \nu_N(\xi,x))},
\end{equation}
for all $(\xi,x) \in \mathcal{X}_\xi \times \mathcal{X}$, where $\nu_N(\xi,x)$ is defined in~\eqref{eqn:learned_dynamics}.

We now proceed with the design of the \rellone-\gp control. As aforementioned, the \rellone-\gp input is composed as
\begin{equation}\label{eqn:RL1GP_input}
u(t) = \hat{u}_{c,\hat{F}}(t) + \hat{u}_{a,\hat{F}}(t).
\end{equation}
Similar to $u_c(t)$ in~\eqref{eqn:CCM_input}, $\hat{u}_{c,\hat{F}}(t)$ is defined as
\begin{equation}\label{eqn:RL1GP:CCM_input}
    u_{c,\hat{F}} = u_d(t) + k_{c,\hat{F}}(x_d(t),x(t)),
\end{equation} where $k_{c,\hat{F}}$ is defined via the analytic solution of the following QP:
\begin{subequations}\label{eqn:RL1GP:CCM_optimization}
\begin{align}
    &k_{c,\hat{F}}(x_d(t),x(t))= \argmin_{k\in\mathbb{R}^m} \norm{k}^2,\\ &\text{s.t.}~2\bar{\gamma}_{\bar{s}}^{\top}(1,t)M(x(t))\dot{x}_k(t)-2\bar{\gamma}_{\bar{s}}^{\top}(0,t)M(x_d(t))\dot{x}_d(t)\leq -2\lambda\mathcal{E}(x_d(t),x(t)),
\end{align}
\end{subequations}
where now $\dot{x}_k(t) = \hat{F}(\xi(t),x(t),u_d(t) + k)$ is compared to $\dot{x}_k(t) = \bar{F}(x(t),u_d(t) + k)$ as in~\eqref{eqn:CCM_optimization}. The incorporation of the learned mean function $\nu_N(\xi,x)$ into the contraction theoretic input is possible due to the fact that the CCM $M(x)$ synthesized for the nominal dynamics $\bar{F}$ satisfying the conditions in~\eqref{eqn:CCM_conditions} is also a valid CCM for both the uncertain dynamics $F$ and the learned dynamics $\hat{F}$~\cite[Lemma~1]{lopez2020adaptive}. The implication is that the CCM $M(x)$ does not need to be re-synthesized whenever the model is updated. To be precise, this property of the CCM $M(x)$ holds because the condition in~\eqref{eqn:killing} implies that the vectors $[b]_{\cdot,j}(x)$ form a Killing vector field for $M(x)$.

The $\mathcal{L}_1$ adaptive input $\hat{u}_{a,\hat{F}}(t)$ once again consists of the state-predictor, adaptation law, and the control law. However, as aforementioned, the input $\hat{u}_{a,\hat{F}}(t)$ now compensates for the remainder uncertainty $h - \nu_N$. In order to re-define the $\hat{u}_{a,\hat{F}}(t)$ control input, we use the constants $\Delta_{\hat{h}}$, $\Delta_{\hat{h}_x}$, and $\Delta_{\hat{h}_\xi}$ defined in~\eqref{eqn:GP_bounds_new}.
The \rellone-\gp \textit{state-predictor} is redesigned as
\begin{equation}\label{eqn:RL1GP:predictor}
    \dot{\hat{x}}(t) = \hat{F}(\xi(t),x(t),u_{c,\hat{F}}(t) + u_{a,\hat{F}}(t) + \hat{\mu}(t)) + A_m \tilde{x}(t), \quad \hat{x}(0) = x_0,
\end{equation} where the learned dynamics $\hat{F}$ are presented in~\eqref{eqn:learned_dynamics}.  Here, $A_m \in \mathbb{R}^{n \times n}$ is any Hurwitz matrix as defined in~\eqref{eqn:state-predictor}.
The inclusion of learned models within the predictor of the $\mathcal{L}_1$ architecture has been previously explored in~\cite{gahlawat2020l1}.
The \textit{adaptation law} is similarly redesigned as
\begin{equation}\label{eqn:RL1GP:adaptation_law}
\dot{\hat{\mu}}(t) = \Gamma \text{Proj}_{\hat{\mathcal{H}}} \left(\hat{\mu}(t), -B^\top (x) P \tilde{x}(t)   \right), \quad \hat{\mu}(0) \in \hat{\mathcal{H}},
\end{equation}
where, as in~\eqref{eqn:adaptation_law}, $\Gamma > 0$ is the adaptation rate and $\tilde{x}(t) = \hat{x}(t) - x(t)$, but now the projection operator is defined on the set $\hat{\mathcal{H}} = \{y \in \mathbb{R}^m~|~\norm{y} \leq \Delta_{\hat{h}}\}$ with $\Delta_{\hat{h}}$ defined in~\eqref{eqn:GP_bounds_new}, instead of $\mathcal{H} = \{y \in \mathbb{R}^m~|~\norm{y} \leq \Delta_h\}$.
Finally, the \textit{control law} is defined as
\begin{equation}\label{eqn:RL1GP:control_law}
u_a(s) = -C(s)\hat{\mu}(s),
\end{equation} where, as in~\eqref{eqn:control_law}, $C(s)$ is a low-pass filter with bandwidth $\omega$ and satisfies $C(0) = \mathbb{I}_m$.

Analogously to conditions in~\eqref{eqn:RL1_conditions} for the filter bandwidth $\omega$ and adaptation rate $\Gamma$, we need to redefine these conditions for the learned representation of the uncertain dynamics in~\eqref{eqn:learned_dynamics}. We begin by defining a few constants analogous to the ones presented in~\eqref{eqn:RL1_bounds}. Consider the positive scalars $\rho_r$ and $\rho_a$ presented in~\eqref{eqn:tube_def}, using which we define
\begin{subequations}\label{eqn:RL1GP_bounds}
\begin{align}
\label{eqn:RL1GP:bounds:M_x}
\Delta_{M_x} &:= \sup_{x \in \mathcal{O}_{x_d}(\rho)} \sum_{i=1}^n \norm{\pdv{M}{x_i}(x)}, \\
\label{eqn:RL1GP:bounds:psi_x}
\Delta_{\Psi_x} &:= 2\Delta_{B_x} + \frac{\Delta_B \Delta_{M_x}}{\ualpha},\\
\label{eqn:RL1GP:bounds:delta_u}
\Delta_{\delta_u} &:= \frac{1}{2}\sup_{x \in \mathcal{O}_{x_d}(\rho)}\left(\frac{\olambda(L^{-\top}(x)Z(x)L^{-1}(x))}{\usigma_{> 0}(B^\top(x)L^{-1}(x))}\right),\\
\Delta_{\hat{f}} &:= \Delta_f + \sup_{(\xi,x) \in \mathcal{X}_\xi \times \mathcal{O}_{x_d}} \norm{B(x)\nu_N(\xi,x)},\\
 \label{eqn:RL1GP:bounds:dx_r}
\Delta_{\dot{x}_r} &:= \Delta_{\hat{f}} + \Delta_B(\norm{\mathbb{I}_m - C(s)}_{\mathcal{L}_1}\Delta_{\hat{h}} + \Delta_{u_d} + \rho\Delta_{\delta_u}), \\
\label{eqn:RL1GP:bounds:dx}
\Delta_{\dot{x}} &:= \Delta_{\hat{f}} + \Delta_B(2\Delta_{\hat{h}} + \Delta_{u_d} + \rho\Delta_{\delta_u}),\\
\label{eqn:RL1GP:bounds:tx}
\Delta_{\tilde{x}} &:= \sqrt{\frac{4\olambda(P)\Delta_{\hat{h}}(\Delta_{\hat{h}_\xi} + \Delta_{\hat{h}_x}\Delta_{\dot{x}})}{\ulambda(P)\underline{\lambda}(Q)} + \frac{4\Delta_{\hat{h}}^2}{\underline{\lambda}(P)}}, \\
\label{eqn:RL1GP:bounds:teta}
\Delta_{\tilde{\eta}} &:= \left(\Delta_{B^\dagger_x} \Delta_{\dot{x}} + (\norm{sC(s)}_{\mathcal{L}_1} + \norm{A_m}) \Delta_{B^\dagger} \right)\Delta_{\tilde{x}},\\
\label{eqn:RL1GP:bounds:theta}
\Delta_\theta &:= \frac{\Delta_B \overline{\alpha}  \Delta_{\tilde{\eta}}}{\lambda},\\
\label{eqn:RL1GP:bounds:dpsi}
\Delta_{\dot{\Psi}} &:= \oalpha\left(
    \Delta_B \Delta_{\dot{\ogamma}_{\bar{s}}} + \frac{\Delta_B \Delta_{M_x} \Delta_{\dot{x}}}{\sqrt{\oalpha\ualpha}} + \Delta_{B_x} \Delta_{\dot{x}}\right), \\
\label{eqn:RL1GP:bounds:dgamma_s}
\Delta_{\dot{\ogamma}_s} &:= \sqrt{\frac{\oalpha}{\ualpha}}\left(\Delta_{\hat{f}_x} + (\Delta_{\hat{h}} + \Delta_{u_d
} + \rho \Delta_{\delta_u})\Delta_{b_x} + \left(\Delta_{\hat{h}_x} + \frac{\sqrt{\ualpha}\Delta_{\delta_u}}{\sqrt{\oalpha}}\right) \Delta_B \right),\\
\Delta_{\hat{f}_x} &:= \Delta_{f_x} + \sup_{(\xi,x) \in \mathcal{X}_\xi \times \mathcal{O}_{x_d}} \norm{\frac{\partial B(x)\nu_N(\xi,x)}{\partial x}},
\end{align}
\end{subequations}
where, as in~\eqref{eqn:RL1_bounds}, $\mathcal{O}_{x_d}(\rho)$ is defined in~\eqref{eqn:tube}; $\Delta_{u_d}$, $\Delta_f$, $\Delta_{f_x}$, $\Delta_B$, $\Delta_{B_x}$, $\Delta_{b_x}$, $\Delta_h$, $\Delta_{h_\xi}$, $\Delta_{h_x}$, $\Delta_{B^\dagger}$ and $\Delta_{B_x^\dagger}$ are defined in~\eqref{eqn:model_bounds}; $\oalpha$ and $\ualpha$ are defined in \cref{assmp:CCM}; and $Z(x)$ is defined as
\[
Z(x):=-\partial_{f} W(x)+ 2\Sym{\pdv{f}{x}(x)W(x)} + 2 \lambda W(x),
\]
where $W(x) = M(x)^{-1}$ is referred to as the dual metric and $L(x)^\top L(x) = W(x)$, and these entities are guaranteed to exist due to the positive definiteness of the CCM $M(x)$. Additionally, $\Delta_{\hat{h}}$, $\Delta_{\hat{h}_x}$, and $\Delta_{\hat{h}_\xi}$ are defined in~\eqref{eqn:GP_bounds_new}. We would like to highlight the fact that the constants in~\eqref{eqn:RL1GP_bounds} differ from the ones in~\eqref{eqn:RL1_bounds} in that $\Delta_h$, $\Delta_{h_x}$, $\Delta_{h_\xi}$, $\Delta_f$, and $\Delta_{f_x}$ have been replaced by $\Delta_{\hat{h}}$, $\Delta_{\hat{h}_x}$, $\Delta_{\hat{h}_\xi}$, $\Delta_{\hat{f}}$, and $\Delta_{\hat{f}_x}$, respectively. Also note that the $\Delta_{\delta_u}$ in~\eqref{eqn:RL1GP:bounds:delta_u} remains the same as in~\eqref{eqn:bounds:delta_u} due to the Killing vector field condition in~\eqref{eqn:killing}.

Analogously to conditions in~\eqref{eqn:RL1_conditions} for the filter bandwidth $\omega$ and adaptation rate $\Gamma$, we need to redefine these conditions for the learned representation of the uncertain dynamics in~\eqref{eqn:learned_dynamics}. To construct the safety and performance certificates based on the learned estimates, using~\eqref{eqn:GP_bounds_new} and~\eqref{eqn:RL1GP_bounds}, we redefine the conditions that the filter bandwidth $\omega$ and adaptation rate $\Gamma$ must satisfy
\begin{equation}\label{eqn:RL1GP_conditions}
    \rho_r^2 \geq \frac{\mathcal{E}(x_d(0),x_0)}{\underline{\alpha}} + \zeta_1(\hat{\kappa}_1,\omega), ~\underline{\alpha} > \zeta_2(\hat{\kappa}_2,\omega) + \zeta_3(\hat{\kappa}_3,\omega), ~ \sqrt{\Gamma} > \frac{\hat{\kappa}_4(\Delta_{\hat{h}},\Delta_{\hat{h}_x},\Delta_{\hat{h}_\xi})}{\rho_a ( \underline{\alpha} - \zeta_2(\hat{\kappa}_2,\omega) - \zeta_3(\hat{\kappa}_3,\omega)   )},
\end{equation}
where
\begin{align}\label{eqn:RL1GP_constants}
    \zeta_1(\hat{\kappa}_1,\omega) =& \frac{\hat{\kappa}_1(\Delta_{\hat{h}},\Delta_{\hat{h}_x},\Delta_{\hat{h}_\xi})}{\omega},~ \zeta_2(\hat{\kappa}_2,\omega) = \frac{\hat{\kappa}_2(\Delta_{\hat{h}},\Delta_{\hat{h}_x},\Delta_{\hat{h}_\xi})}{\omega},~ \zeta_3(\hat{\kappa}_3,\omega) = \frac{\hat{\kappa}_3(\Delta_{\hat{h}},\Delta_{\hat{h}_x})}{\omega},
\end{align}
and where
\begin{subequations}\label{eqn:RL1GP:kappas}
\begin{align}
    \kappa_1(\Delta_{\hat{h}},\Delta_{\hat{h}_x},\Delta_{\hat{h}_\xi}) =& 2 \rho \Delta_B \frac{\bar{\alpha}}{\underline{\alpha}} \left(
    \frac{\Delta_{\hat{h}}}{|2\lambda/\omega - 1|} + \frac{\Delta_{\hat{h}_\xi} + \Delta_{\hat{h}_x} \Delta_{\dot{x}_r}}{2 \lambda},
    \right), \\
    \kappa_2(\Delta_{\hat{h}},\Delta_{\hat{h}_x},\Delta_{\hat{h}_\xi}) =& \bar{\alpha} \Delta_{\Psi_x} \frac{\bar{\alpha}}{\underline{\alpha}} \left(
    \frac{\Delta_{\hat{h}}}{|2\lambda/\omega - 1|} + \frac{\Delta_{\hat{h}_\xi} + \Delta_{\hat{h}_x} \Delta_{\dot{x}_r}}{2 \lambda}
    \right),\\
    \kappa_3(\Delta_{\hat{h}},\Delta_{\hat{h}_x}) =& \bar{\alpha} \Delta_{\hat{h}_x} \left(
    \frac{4 \lambda \Delta_B + \Delta_{\dot{\Psi}}}{\lambda}
    \right),\\
    \kappa_4(\Delta_{\hat{h}},\Delta_{\hat{h}_x},\Delta_{\hat{h}_\xi}) =& \Delta_\theta.
\end{align}
\end{subequations}

Note that these conditions for the \rellone-\gp differ from the conditions for the \rellone control in~\eqref{eqn:RL1_conditions} in that these are defined using the constants $\hat{\kappa}_i$, $i \in \{1,\dots,4\}$, which are in-turn defined using $\Delta_{\hat{h}}$, $\Delta_{\hat{h}_x}$, and $\Delta_{\hat{h}_\xi}$ presented in~\eqref{eqn:GP_bounds_new}. In contrast, the conditions for the \rellone control in~\eqref{eqn:RL1_conditions} are defined using the constants $\kappa_i$, $i \in \{1,\dots,4\}$, that are in turn defined using $\Delta_{h}$, $\Delta_{h_x}$, and $\Delta_{h_\xi}$, which are the known conservative bounds for the uncertainty as presented in Assumptions~\ref{assmp:models:functions} and~\ref{assmp:models:functions_derivatives}. Thus, the only difference between $\hat{\kappa}_i$ and $\kappa_i$, $i \in \{1,\dots,4\}$, is the use of the constants in~\eqref{eqn:GP_bounds} as opposed to the bounds in Assumptions~\ref{assmp:models:functions} and~\ref{assmp:models:functions_derivatives}.
In conclusion, given the posterior distribution in~\eqref{eqn:function:posterior}, the \rellone-\gp control input is defined via~\eqref{eqn:RL1GP_input},~\eqref{eqn:RL1GP:CCM_input},~\eqref{eqn:RL1GP:predictor},~\eqref{eqn:RL1GP:adaptation_law} and~\eqref{eqn:RL1GP:control_law}.


We now state the main result.
\begin{theorem}\label{thm:RL1GP_main}
Let Assumptions~\ref{assmp:models:functions},~\ref{assmp:models:functions_derivatives},~\ref{assmp:GP} and~\ref{assmp:CCM} hold and suppose the uniform bounds in Theorem~\ref{thm:uniform_bounds} are computed for some $\delta \in (0,1)$ and $\tau > 0$ using the posterior distributions in~\eqref{eqn:function:posterior}-\eqref{eqn:function_derivative:posterior}. Furthermore, let the filter bandwidth $\omega$ and rate of adaptation $\Gamma$ satisfy~\eqref{eqn:RL1GP_conditions}, using $\zeta_i(\hat{\kappa}_i,\omega)$, $i \in \{1,\dots,3\}$, presented in~\eqref{eqn:RL1GP_constants}, and defined using the constants in~\eqref{eqn:GP_bounds_new}.

Given any desired state-input pair $(x_d(t),u_d(t))$ satisfying the learned deterministic dynamics
\[
\dot{x}_d(t) = \hat{F}(\xi(t),x_d(t),u_d(t)),
\]
the state of the actual uncertain system
\[
\dot{x}(t) = F(\xi(t),x(t),u(t))
\] with \rellone-\gp control input~\eqref{eqn:RL1GP_input},~\eqref{eqn:RL1GP:CCM_input},~\eqref{eqn:RL1GP:predictor},~\eqref{eqn:RL1GP:adaptation_law} and~\eqref{eqn:RL1GP:control_law}, satisfies, with probability at least $1- \delta$
\begin{equation} \label{eqn:main_theorem:UB}
x(t) \in \Omega(\rho,x_d(t)) \subset \mathcal{O}_{x_d}(\rho), \quad \forall t \geq 0,
\end{equation}
with $\rho$ defined in~\eqref{eqn:tube_def}. Furthermore, the actual system state $x(t)$ is uniformly ultimately bounded, with probability at least $1- \delta$, as
\begin{equation}\label{eqn:main_theorem:UUB}
x(t) \in \Omega(\hat{\delta}(\omega,T),x_d(t)) \subset \Omega(\rho,x_d(t)), \quad  \forall t \geq T > 0,
\end{equation}
where the uniform ultimate bound (UUB) is given by
\[
\hat{\delta}(\omega,T) = \bar{\mu}(\omega,T) + \rho_a, \quad \bar{\mu}(\omega,T) = \sqrt{e^{-2\lambda T}\mathcal{E}(x_d(0),x_0)/ \underline{\alpha}  + \zeta_1(\hat{\kappa}_1,\omega)}.
\]
\end{theorem}
\begin{proof}
Consider the closed-loop system given by~\eqref{eqn:system_dynamics}(\eqref{eqn:learned_dynamics})~\eqref{eqn:RL1GP_input},~\eqref{eqn:RL1GP:CCM_input},~\eqref{eqn:RL1GP:predictor},~\eqref{eqn:RL1GP:adaptation_law} and~\eqref{eqn:RL1GP:control_law}. Under the assumption that the bounds in~\eqref{eqn:GP_bounds} hold, using~\cite[Thm.~5.1]{lakshmanan2020safe}, it can be shown that the state $x(t)$ of the closed-loop system satisfies~\eqref{eqn:main_theorem:UB} and~\eqref{eqn:main_theorem:UUB}.
The proof is then concluded using the fact that by Theorem~\ref{thm:uniform_bounds}, the bounds in~\eqref{eqn:GP_bounds} (defined using~\eqref{eqn:GP_bounds_new}) hold with probability at least $1-\delta$.
\end{proof}

A few crucial comments are in order. As the learning improves, the constants in~\eqref{eqn:GP_bounds_new} decrease, and hence, the constants in~\eqref{eqn:RL1GP_constants} decrease. This fact implies that without changing the filter bandwidth $\omega$ and adaptation rate $\Gamma$, the UUB in Theorem~\ref{thm:RL1GP_main} decreases. The decrease in the UUB, and the lack of a requirement for the re-tuning of the control parameters, is due to the monotonic dependence of the constants $\zeta_i$ on $\hat{\kappa}_i$, $i \in \{1,\dots,3\}$. Furthermore, as aforementioned, the CCM $M(x)$ does not need to be re-synthesized as the model is updated using learning. Thus, without re-tuning the parameters of the \rellone-\gp control input, with the \rellone-\gp control designed using only Assumptions~\ref{assmp:models:functions} and~\ref{assmp:models:functions_derivatives}, the performance improves as a function of learning. We would also highlight the fact that in the absence of learning, the \rellone-\gp control degenerates into the $\mathcal{CL}_1$ control, presented in Section~\ref{sec:prelim}, while still providing apriori computable safety bounds.



\section{Simulation Results}

We provide an illustrative example of a 6-DOF planar quadrotor at different levels during the learning process using a variety of motion planners. The dynamics of the vehicle can be expressed in the following control-affine form following \cite{singh2019robust}:
\[
    \begin{bmatrix}
    p_x \\ p_z \\ \theta \\ v_x \\ v_z \\ \dot{\theta}
    \end{bmatrix} =
    \begin{bmatrix}
        v_x\cos(\theta) - v_z\sin(\theta) \\
        v_x\sin(\theta) + v_z\cos(\theta) \\
        \dot{\theta} \\
        v_z\dot{\theta} - g \cos(\theta) \\
        -v_x\dot{\theta} - g \sin(\theta) \\
        0
    \end{bmatrix} +
    \begin{bmatrix}
    0 & 0 \\
    0 & 0 \\
    0 & 0 \\
    0 & 0 \\
    1 & 0 \\
    0 & 1
    \end{bmatrix} \begin{bmatrix}
    u_F \\
    u_M
    \end{bmatrix},
\]
\begin{figure}[t]
    \centering
    \includegraphics[width=0.5\textwidth]{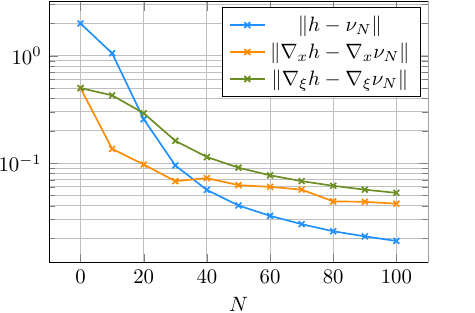}
    \caption{The decay of the uncertainty bounds based on the high-probability prediction-error bounds from \cref{eqn:GP_bounds_new} with the growth of the dataset.}
    \label{fig:sims_learning}
\end{figure}
where $p_x, p_z$ is the position of the quadrotor in the $x-z$ plane; $v_x, v_z$ are the velocities of the quadrotor in the body frame; $\theta, \dot{\theta}$ are the pitch angle and rate; $g$ is the gravitational constant; and $u_F$ and $u_M$ are the thrust and moment control inputs respectively. Additionally, the planar quadrotor is required to always meet the following state constraints:
\begin{subequations}\label{eqn:sims:constraints}
\begin{align*}
    -\frac{\pi}{4} \le & \theta \le \frac{\pi}{4}, \qquad -2 \le  v_x \le 2,  \\
    -\frac{\pi}{3} \le & \dot{\theta} \le \frac{\pi}{3}, \qquad -1 \le v_z \le 1.
\end{align*}
\end{subequations}
The contraction metric is synthesized using a sum-of-squares programming approach described in \citep{singh2019robust}. In the following examples, we consider that the unmodeled uncertainty is given by
\[
h(t, x) = \begin{bmatrix} -1 -0.1(v_x^2 + v_y^2) \\ 0.3\cos(t)\end{bmatrix}.
\]
The first component of the uncertainty affects the total thrust and is indicative of an off-trim control and drag-like parasitic force, whereas the second component is a time-varying disturbance that is injected into the moment input channel. Recall that the time-varying parameter from \cref{eqn:system_dynamics} is simply $\xi(t) = t$. In each of the examples we show the evolution of the safety guarantees across three learning episodes and the resulting improvement in performance and optimality.
The dataset is generated by using Latin hypercube sampling \cite{mackay1979comparison,urquhart2020surrogate} across the state space, but one could also use sophisticated exploration techniques to safely gather data based on our framework. Prior to learning, the bounds on the uncertainty and its growth over the state-space are conservatively estimated as
\begin{figure}[t]
\centering
\subfloat[][Episode 1: $\omega = 90 \ \mathrm{rad/s}, \Gamma = 7\mathrm{e}{10}$]{
\includegraphics[width=0.5\textwidth]{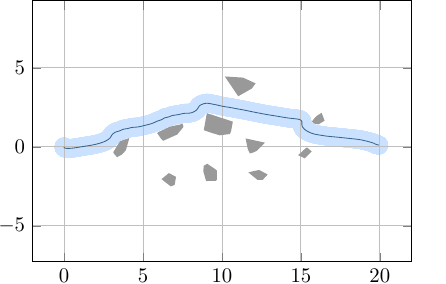}
\label{fig:forest_wide}}
\subfloat[][Episode 2: $\omega = 30 \ \mathrm{rad/s}, \Gamma = 2\mathrm{e}{6}$]{
\includegraphics[width=0.5\textwidth]{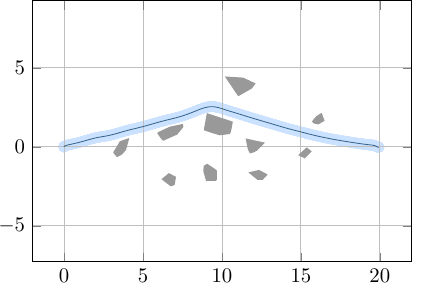}
\label{fig:forest_medium}}\\
\subfloat[][Episode 3: $\omega = 30 \ \mathrm{rad/s}, \Gamma = 2\mathrm{e}{6}$]{
\includegraphics[width=0.5\textwidth]{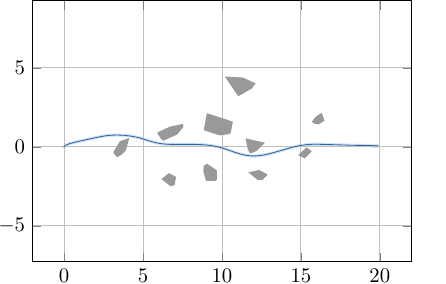}
\label{fig:forest_narrow}}
\caption{Planar quadrotor flight across an obstacle forest with (a) only a deterministic knowledge of the uncertainty, (b) model learned with $N=25$ dataset, (c) model learned with $N=100$ dataset.}
\label{fig:forest}
\end{figure}

\[
\Delta_h = 2.0, \quad \Delta_{h_x} = 0.5, \quad \Delta_{h_{\xi}} = 0.5.
\]
One can analytically verify that the true bounds of the uncertainty are indeed lower than our estimates:
\[
\norm{h(t,x)} = 1.53, \quad \norm{\pdv{h}{x}(t,x)} = 0.45, \quad \norm{\pdv{h}{t}(t,x)} = 0.3, \quad \forall x \in \mathcal{X} \textrm{ and } t \ge 0,
\]
where $\mathcal{X}$ is defined using the state contraint set in \cref{eqn:sims:constraints}.
In the first two examples, we begin with a simplified architecture of the \rellone-\gp input where the learned estimates are only used to improve performance and optimality with respect to the nominal model. That is, the planner does not incorporate the learned updates. Later in third example, we provide a sim for the complete \rellone-\gp architecture but in a simplified environment. We consider three instances/episodes during the learning transients with $N=0$, $N=25$ and $N=100$ samples. The corresponding uncertainty bounds for each episode are shown in \Cref{fig:sims_learning}. The results presented use $\delta = 0.1$ and $\tau = 1\mathrm{e}{-8}$ for the terms defined in \cref{thm:uniform_bounds}, therefore the performance bounds indicated in the figures hold with probability at least $0.9$. The examples were simulated using the Julia programming language \cite{bezanson2017julia} and the Pluto reactive environment \cite{fons2020pluto}.

\begin{figure}[t]
\centering
\subfloat[][Episode 1: $\omega = 90 \ \mathrm{rad/s}, \Gamma = 7\mathrm{e}{10}$]{
\includegraphics[width=0.5\textwidth]{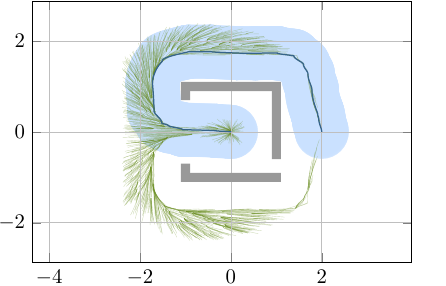}
\label{fig:trap_wide}}
\subfloat[][Episode 2: $\omega = 30 \ \mathrm{rad/s}, \Gamma = 2\mathrm{e}{6}$]{
\includegraphics[width=0.5\textwidth]{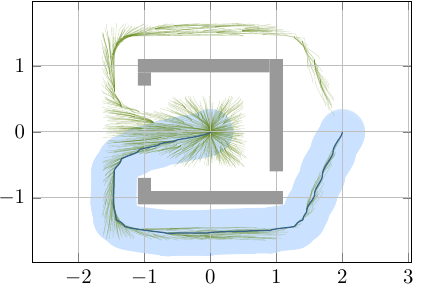}
\label{fig:trap_medium}}\\
\subfloat[][Episode 3: $\omega = 30 \ \mathrm{rad/s}, \Gamma = 2\mathrm{e}{6}$]{
\includegraphics[width=0.5\textwidth]{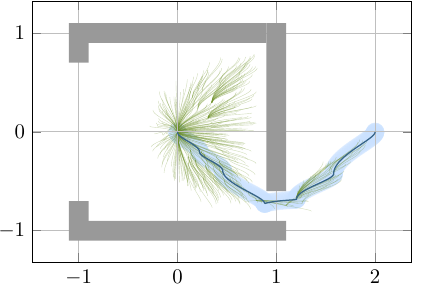}
\label{fig:trap_narrow}}
\caption{Planar quadrotor escaping a bug trap using (a) only a deterministic knowledge of the uncertainty, (b) model learned with $N=25$ dataset, (c) model learned with $N=100$ dataset. The green lines indicate the edges of the random geometric graph constructed by BIT*.}
\label{fig:trap}
\end{figure}

\begin{figure}
\centering
\subfloat[][Episode 1: $\omega = 90 \ \mathrm{rad/s}, \Gamma = 7\mathrm{e}{10}$]{
\includegraphics[width=\textwidth]{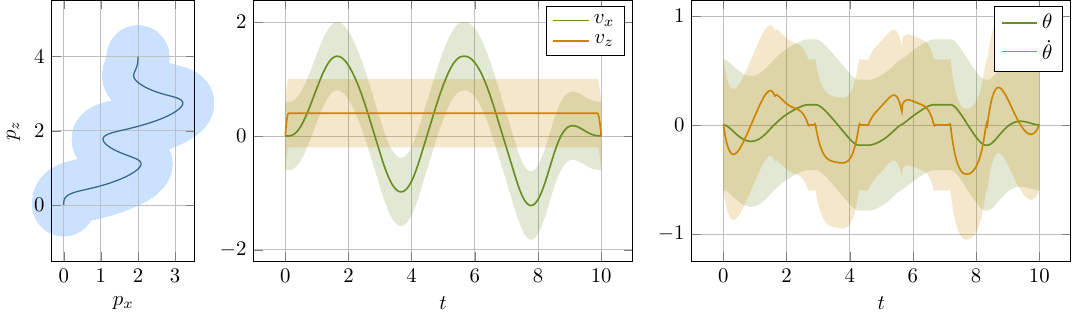}
\label{fig:optim_wide}} \\
\subfloat[][Episode 2: $\omega = 30 \ \mathrm{rad/s}, \Gamma = 2\mathrm{e}{6}$]{
\includegraphics[width=\textwidth]{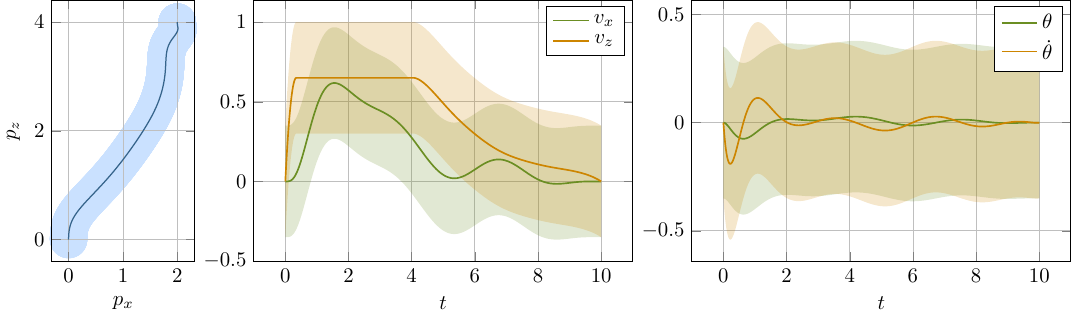}
\label{fig:optim_medium}} \\
\subfloat[][Episode 3: $\omega = 30 \ \mathrm{rad/s}, \Gamma = 2\mathrm{e}{6}$]{
\includegraphics[width=\textwidth]{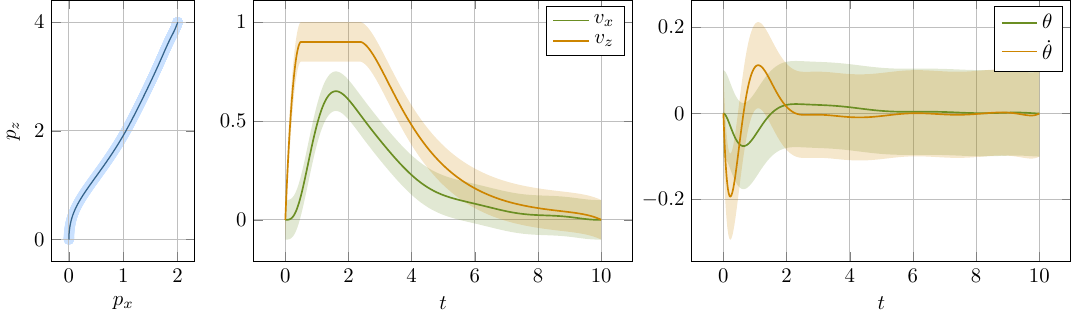}
\label{fig:optim_narrow}}
\caption{Planar quadrotor traveling from $(0,0)$ to $(2,4)$ using (a) only a deterministic knowledge of the uncertainty, (b) model learned with $N=25$ dataset, (c) model learned with $N=100$ dataset. On the left, the plots show the navigation and the performance bounds for the vehicle on the $x-z$ plane. The plots in the middle and the right show the state trajectories and their corresponding performance bounds.}
\label{fig:optim}
\end{figure}

\begin{example}[Obstacle Forest]
\label{ex:forest}
The quadrotor is tasked to safely fly across a forest of convex polygonal obstacles from the origin to position 20 meters away while minimizing the following discrete-time objective
\[
J = \sum_{k=0}^{T-1} x_k^\top Q x_k + (x_T - x_{\textrm{goal}})^\top Q_f (x_T - x_{\textrm{goal}}),
\]
where $x_k$ is the state at the $k^{\textrm{th}}$ time-instant, $x_\textrm{goal}$ is the goal state, and $Q$ and $Q_f$ are positive definite diagonal matrices. In this example, we use MPPI~\citep{williams2018information} to generate the feasible trajectories based on the pre-computed tube size.  MPPI was configured to generate 500 trajectory rollouts at a frequency of 50 Hz with a prediction horizon of 2 seconds. The tubes depicted in \Cref{fig:forest} are only a projection of the tube $\mathcal{O}_{x_d}$ onto the vehicle position but they also extend into the rest of the state-space limiting the overall maneuverability of the quadrotor. For instance when the tube-size is $\rho = 0.6$, the maximum pitch angle is  approximately $\pm$ 11 degrees instead of the full $\pm$ 45 degrees pitch that the contraction metric was initially designed for.  Initially in \Cref{fig:forest_wide}, the model knowledge is poor and the tubes guaranteed by $\mathcal{CL}_1$ control are conservative based only on the deterministic knowledge of the uncertainty, with a tube size of $\rho = 0.6$. This lack of knowledge results in a circuitous path that takes over about 27~seconds for the vehicle to safely traverse. As the data is incorporated into the model, the performance improvement can be seen in~\Cref{fig:forest}. The trajectory shown in~\Cref{fig:forest_medium} has a tube radius of $\rho = 0.35$  and has a duration of 16~seconds. In~\Cref{fig:forest_narrow}, the trajectory has a tube radius of $\rho = 0.1$, and the vehicle can navigate the environment to the final position in only 14~seconds. Note that after incorporating the learned model, both the $\mathcal{L}_1$ filter bandwidth and the adaptation rate are reduced to improve the robustness margin of the closed-loop system and lower the computational burden of the controller.
\end{example}

\begin{example}[Bug-Trap]
In this example, the quadrotor must safely escape a box trap from the origin and arrive at a point on the other side of the trap. For such problems, complete or probabilistically complete planners are the algorithms of choice since other methods typically get stuck at a local minimum and never reach the goal. We use the popular sampling-based planner BIT* \cite{gammell2015batch} with the two-point boundary value problem solved using ALTRO \cite{howell2019altro}. For the sake of simplicity, our implementation of BIT* only samples in the position space and the remaining states are assumed to be zero at each sample, but this can be relaxed if the planner is constructed following the approach described in \cite{xie2015toward}. BIT* is configured with a batch-size of 500 samples and a total of 10 batches. Similar to \cref{ex:forest}, each of the simulations in \Cref{fig:trap} show the safe navigation using the tube bounds during different instances of the learning process.
\end{example}

\begin{example}[Improving Optimality]
In the previous two examples, the learned model was simply used to compensate for the uncertainty and was not explicitly used to generate desired trajectories that exploit the newly learned model. In this example, the trajectory optimization solver \cite{howell2019altro} uses the mean dynamics of the GP predictive function from \cref{eqn:function:posterior} to improve the quality of the solution. The quadrotor is tasked to fly from the origin to $(2,4)$ in 10 seconds while minimizing the following LQR objective:
\[
J = \sum_{k=0}^{T-1} \left(x_k^\top Q x_k + u_k^\top R u_k \right) + (x_T - x_{\textrm{goal}})^\top Q_f (x_T - x_{\textrm{goal}}),
\]
where $x_k$ and $u_k$ are the state and controls at the $k^{\textrm{th}}$ time-instant, $x_\textrm{goal}$ is the goal state, and $Q$, $R$ and $Q_f$ are positive definite diagonal matrices. Only the state constraints from \cref{eqn:sims:constraints} are active and no other obstacles are present so that we can clearly see the improvement in optimality. In \Cref{fig:optim_wide}, the vehicle can only reach a maximum of $0.4$ m/s in the body $z-$axis and must therefore exploit the remaining maneuverability in its body $x-$axis to fly to the goal location. This results in a zig-zag flight path with large oscillations in the vehicle pitch. As the learned model is incorporated in \Cref{fig:optim_medium,fig:optim_narrow}, the solver arrives at smoother solutions which don't oscillate as much as the first episode. In \Cref{fig:optim_narrow}, the vehicle is capable of reaching much faster speeds in its body $z-$axis and therefore plans a much more straightforward path to the goal.
\end{example}


\section{Conclusion}
In this work, we have presented the \rellone-\gp framework, which enables safe simultaneous learning and control.
The safety of the method is certified by the tracking error bounds produced by the ancillary $\mathcal{CL}_1$ controller.
The learning is performed using Gaussian process regression.
The learned Gaussian process model can be used to generate high probability uniform error bounds, which are incorporated into the controller to improve the tracking error bounds.
Future work will extend the architecture to leverage the tracking error bounds in the path planning phase. The bounds are used to ensure safety, but can also be extended to provide worst case estimates for both the uncertainty reduction and cost associated with a desired trajectory.
Finally, the guarantees will be extended to a larger class of nonlinear systems, explored in output feedback formulation, and other possible generalizations.


\section*{Acknowledgments}

This work is financially supported by the National Aeronautics and Space Administration (NASA), Air Force Office of Scientific Research (AFOSR), National Science Foundation (NSF) Cyber Physical Systems (CPS) award \# 1932529, and NSF National Robotics Initiative 2.0 (NRI-2.0) award \# 1830639.


\bibliographystyle{abbrvnat}
\bibliography{ref}

\begin{thebibliography}{50}
\providecommand{\natexlab}[1]{#1}
\providecommand{\url}[1]{\texttt{#1}}
\expandafter\ifx\csname urlstyle\endcsname\relax
  \providecommand{\doi}[1]{doi: #1}\else
  \providecommand{\doi}{doi: \begingroup \urlstyle{rm}\Url}\fi

\bibitem[Ames et~al.(2016)Ames, Xu, Grizzle, and Tabuada]{ames2016control}
A.~D. Ames, X.~Xu, J.~W. Grizzle, and P.~Tabuada.
\newblock Control barrier function based quadratic programs for safety critical
  systems.
\newblock \emph{IEEE Transactions on Automatic Control}, 62\penalty0
  (8):\penalty0 3861--3876, 2016.

\bibitem[Ames et~al.(2019)Ames, Coogan, Egerstedt, Notomista, Sreenath, and
  Tabuada]{ames2019control}
A.~D. Ames, S.~Coogan, M.~Egerstedt, G.~Notomista, K.~Sreenath, and P.~Tabuada.
\newblock Control barrier functions: Theory and applications.
\newblock In \emph{Proceedings of 18th European Control Conference}, pages
  3420--3431, Naples, Italy, 2019.

\bibitem[Aswani et~al.(2013)Aswani, Gonzalez, Sastry, and
  Tomlin]{aswani2013provably}
A.~Aswani, H.~Gonzalez, S.~S. Sastry, and C.~Tomlin.
\newblock Provably safe and robust learning-based model predictive control.
\newblock \emph{Automatica}, 49\penalty0 (5):\penalty0 1216--1226, 2013.

\bibitem[Berkenkamp et~al.(2017)Berkenkamp, Turchetta, Schoellig, and
  Krause]{berkenkamp2017safe}
F.~Berkenkamp, M.~Turchetta, A.~Schoellig, and A.~Krause.
\newblock Safe model-based reinforcement learning with stability guarantees.
\newblock In \emph{Proceedings of 31st Conference on Neural Information
  Processing Systems}, pages 908--918, Long Beach, CA, USA, 2017.

\bibitem[Bezanson et~al.(2017)Bezanson, Edelman, Karpinski, and
  Shah]{bezanson2017julia}
J.~Bezanson, A.~Edelman, S.~Karpinski, and V.~B. Shah.
\newblock Julia: {A} fresh approach to numerical computing.
\newblock \emph{SIAM review}, 59\penalty0 (1):\penalty0 65--98, 2017.

\bibitem[Bishop(2006)]{bishop2006pattern}
C.~M. Bishop.
\newblock \emph{Pattern recognition and machine learning}.
\newblock Springer, 2006.

\bibitem[Chow et~al.(2018)Chow, Nachum, Duenez-Guzman, and
  Ghavamzadeh]{chow2018lyapunov}
Y.~Chow, O.~Nachum, E.~Duenez-Guzman, and M.~Ghavamzadeh.
\newblock A {L}yapunov-based approach to safe reinforcement learning.
\newblock In \emph{Proceedings of 32nd Conference on Neural Information
  Processing Systems}, pages 8092--8101, Quebec, Canada, 2018.

\bibitem[Cichella et~al.(2017)Cichella, Kaminer, Walton, and
  Hovakimyan]{cichella2017optimal}
V.~Cichella, I.~Kaminer, C.~Walton, and N.~Hovakimyan.
\newblock Optimal motion planning for differentially flat systems using
  {B}ernstein approximation.
\newblock \emph{IEEE Control Systems Letters}, 2\penalty0 (1):\penalty0
  181--186, 2017.

\bibitem[Fridovich-Keil et~al.(2018)Fridovich-Keil, Herbert, Fisac, Deglurkar,
  and Tomlin]{fridovich2018planning}
D.~Fridovich-Keil, S.~L. Herbert, J.~F. Fisac, S.~Deglurkar, and C.~J. Tomlin.
\newblock Planning, fast and slow: {A} framework for adaptive real-time safe
  trajectory planning.
\newblock In \emph{Proceedings of IEEE International Conference on Robotics and
  Automation}, pages 387--394, Brisbane, Australia, 2018.

\bibitem[Gahlawat et~al.(2020)Gahlawat, Zhao, Patterson, Hovakimyan, and
  Theodorou]{gahlawat2020l1}
A.~Gahlawat, P.~Zhao, A.~Patterson, N.~Hovakimyan, and E.~Theodorou.
\newblock $\mathcal{L}_1$-$\mathcal{GP}$: $\mathcal{L}_1$ adaptive control with
  {B}ayesian learning.
\newblock In \emph{Proceedings of 2nd Learning for Dyanmics \& Control}, volume
  120, pages 1--12, 2020.
\newblock Online.

\bibitem[Gammell et~al.(2015)Gammell, Srinivasa, and Barfoot]{gammell2015batch}
J.~D. Gammell, S.~S. Srinivasa, and T.~D. Barfoot.
\newblock Batch informed trees ({BIT}*): {S}ampling-based optimal planning via
  the heuristically guided search of implicit random geometric graphs.
\newblock In \emph{Proceedings of IEEE International Conference on Robotics and
  Automation}, pages 3067--3074, Seattle, WA, USA, 2015. IEEE.

\bibitem[Herbert et~al.(2017)Herbert, Chen, Han, Bansal, Fisac, and
  Tomlin]{herbert2017fastrack}
S.~L. Herbert, M.~Chen, S.~Han, S.~Bansal, J.~F. Fisac, and C.~J. Tomlin.
\newblock Fa{ST}rack: {A} modular framework for fast and guaranteed safe motion
  planning.
\newblock In \emph{Proceedings of 56th IEEE Conference on Decision and
  Control}, pages 1517--1522, Melbourne, Australia, 2017.

\bibitem[Hewing et~al.(2019)Hewing, Kabzan, and Zeilinger]{hewing2019cautious}
L.~Hewing, J.~Kabzan, and M.~N. Zeilinger.
\newblock Cautious model predictive control using {G}aussian process
  regression.
\newblock \emph{IEEE Transactions on Control Systems Technology}, 2019.
\newblock Early Access.

\bibitem[Hovakimyan and Cao(2010)]{hovakimyan2010L1}
N.~Hovakimyan and C.~Cao.
\newblock \emph{$\mathcal{L}_1$ Adaptive Control Theory: {G}uaranteed
  Robustness with Fast Adaptation}.
\newblock SIAM, Philadelphia, PA, USA, 2010.

\bibitem[Howell et~al.(2019)Howell, Jackson, and Manchester]{howell2019altro}
T.~A. Howell, B.~E. Jackson, and Z.~Manchester.
\newblock {ALTRO}: {A} fast solver for constrained trajectory optimization.
\newblock In \emph{Proceedings of IEEE/RSJ International Conference on
  Intelligent Robots and Systems}, pages 7674--7679, Macau, China, 2019.

\bibitem[Khalil(2014)]{khalil2014nonlinear}
H.~K. Khalil.
\newblock \emph{Nonlinear control}.
\newblock Pearson, London, UK, 2014.

\bibitem[Knight(2002)]{knight2002safety}
J.~C. Knight.
\newblock Safety critical systems: {C}hallenges and directions.
\newblock In \emph{Proceedings of the 24th International Conference on Software
  Engineering}, pages 547--550, Orlando, FL, USA, 2002.

\bibitem[Koller et~al.(2018)Koller, Berkenkamp, Turchetta, and
  Krause]{koller2018learning}
T.~Koller, F.~Berkenkamp, M.~Turchetta, and A.~Krause.
\newblock Learning-based model predictive control for safe exploration.
\newblock In \emph{Proceedings of 57th IEEE Conference on Decision and
  Control}, pages 6059--6066, Miami Beach, FL, USA, 2018.

\bibitem[Lakshmanan et~al.(2020)Lakshmanan, Gahlawat, and
  Hovakimyan]{lakshmanan2020safe}
A.~Lakshmanan, A.~Gahlawat, and N.~Hovakimyan.
\newblock Safe feedback motion planning: A contraction theory and
  $\mathcal{L}_1$-adaptive control based approach.
\newblock \emph{arXiv:2004.01142}, 2020.

\bibitem[LaValle and Kuffner~Jr(2001)]{lavalle2001randomized}
S.~M. LaValle and J.~J. Kuffner~Jr.
\newblock Randomized kinodynamic planning.
\newblock \emph{The International Journal of Robotics Research}, 20\penalty0
  (5):\penalty0 378--400, 2001.

\bibitem[Lederer et~al.(2019)Lederer, Umlauft, and Hirche]{lederer2019uniform}
A.~Lederer, J.~Umlauft, and S.~Hirche.
\newblock Uniform error bounds for {G}aussian process regression with
  application to safe control.
\newblock In \emph{Proceedings of 33rd Conference on Neural Information
  Processing Systems}, pages 659--669, Vancouver, BC, Canada, 2019.

\bibitem[Lopez and Slotine(2020)]{lopez2020adaptive}
B.~T. Lopez and J.-J.~E. Slotine.
\newblock Adaptive nonlinear control with contraction metrics.
\newblock \emph{IEEE Control Systems Letters}, 5\penalty0 (1):\penalty0
  205--210, 2020.

\bibitem[Lopez et~al.(2019)Lopez, Howl, and Slotine]{lopez2019dynamic}
B.~T. Lopez, J.~P. Howl, and J.-J.~E. Slotine.
\newblock Dynamic tube {MPC} for nonlinear systems.
\newblock In \emph{Proceedings of American Control Conference}, pages
  1655--1662, Philadelphia, PA, 2019.

\bibitem[Lopez et~al.(2020)Lopez, Slotine, and How]{lopez2020robust}
B.~T. Lopez, J.-J.~E. Slotine, and J.~P. How.
\newblock Robust adaptive control barrier functions: An adaptive \& data-driven
  approach to safety.
\newblock \emph{IEEE Control Systems Letters}, 2020.
\newblock Early Access.

\bibitem[Magni et~al.(2001)Magni, Nijmeijer, and Van
  Der~Schaft]{magni2001receding}
L.~Magni, H.~Nijmeijer, and A.~Van Der~Schaft.
\newblock A receding--horizon approach to the nonlinear $\mathcal{H}_\infty$
  control problem.
\newblock \emph{Automatica}, 37\penalty0 (3):\penalty0 429--435, 2001.

\bibitem[Manchester and Slotine(2017)]{manchester2017control}
I.~R. Manchester and J.-J.~E. Slotine.
\newblock Control contraction metrics: {C}onvex and intrinsic criteria for
  nonlinear feedback design.
\newblock \emph{IEEE Transactions on Automatic Control}, 62\penalty0
  (6):\penalty0 3046--3053, 2017.

\bibitem[McKay et~al.(1979)McKay, Beckman, and Conover]{mackay1979comparison}
M.~D. McKay, R.~J. Beckman, and W.~J. Conover.
\newblock A comparison of three methods for selecting values of input variables
  in the analysis of output from a computer code.
\newblock \emph{Technometrics}, 21\penalty0 (2):\penalty0 239--245, 1979.

\bibitem[Ostafew et~al.(2016)Ostafew, Schoellig, and
  Barfoot]{ostafew2016robust}
C.~J. Ostafew, A.~P. Schoellig, and T.~D. Barfoot.
\newblock Robust constrained learning-based {NMPC} enabling reliable mobile
  robot path tracking.
\newblock \emph{The International Journal of Robotics Research}, 35\penalty0
  (13):\penalty0 1547--1563, 2016.

\bibitem[Pereida and Schoellig(2018)]{pereida2018adaptive}
K.~Pereida and A.~P. Schoellig.
\newblock Adaptive model predictive control for high-accuracy trajectory
  tracking in changing conditions.
\newblock In \emph{Proceedings of IEEE/RSJ International Conference on
  Intelligent Robots and Systems}, pages 7831--7837, Madrid, Spain, 2018.

\bibitem[Perkins and Barto(2002)]{perkins2002lyapunov}
T.~J. Perkins and A.~G. Barto.
\newblock Lyapunov design for safe reinforcement learning.
\newblock \emph{Journal of Machine Learning Research}, 3\penalty0
  (12):\penalty0 803--832, 2002.

\bibitem[Pravitra et~al.(2020)Pravitra, Ackerman, Cao, Hovakimyan, and
  Theodorou]{pravitra2020l1}
J.~Pravitra, K.~A. Ackerman, C.~Cao, N.~Hovakimyan, and E.~A. Theodorou.
\newblock $\mathcal{L}_1$-adaptive {MPPI} architecture for robust and agile
  control of multirotors.
\newblock \emph{arXiv:2004.00152}, 2020.

\bibitem[Raimondo et~al.(2009)Raimondo, Limon, Lazar, Magni, and
  Camacho]{raimondo2009min}
D.~M. Raimondo, D.~Limon, M.~Lazar, L.~Magni, and E.~F. Camacho.
\newblock Min-max model predictive control of nonlinear systems: A unifying
  overview on stability.
\newblock \emph{European Journal of Control}, 15\penalty0 (1):\penalty0 5--21,
  2009.

\bibitem[Rakovi{\'c} et~al.(2016)Rakovi{\'c}, Levine, and
  A{\c{c}}ikmese]{rakovic2016elastic}
S.~V. Rakovi{\'c}, W.~S. Levine, and B.~A{\c{c}}ikmese.
\newblock Elastic tube model predictive control.
\newblock In \emph{Proceedings of American Control Conference}, pages
  3594--3599, Boston, MA, USA, 2016.

\bibitem[Recht(2019)]{recht2019tour}
B.~Recht.
\newblock A tour of reinforcement learning: {T}he view from continuous control.
\newblock \emph{Annual Review of Control, Robotics, and Autonomous Systems},
  2:\penalty0 253--279, 2019.

\bibitem[Rosolia and Borrelli(2019)]{rosolia2019sample}
U.~Rosolia and F.~Borrelli.
\newblock Sample-based learning model predictive control for linear uncertain
  systems.
\newblock In \emph{Proceedings of 58th IEEE Conference on Decision and
  Control}, pages 2702--2707, Nice, France, 2019.

\bibitem[Singh et~al.(2017)Singh, Majumdar, Slotine, and
  Pavone]{singh2017robust}
S.~Singh, A.~Majumdar, J.-J. Slotine, and M.~Pavone.
\newblock Robust online motion planning via contraction theory and convex
  optimization.
\newblock In \emph{Proceedings of IEEE International Conference on Robotics and
  Automation}, pages 5883--5890, Marina Bay Sands, Singapore, 2017.

\bibitem[Singh et~al.(2019)Singh, Landry, Majumdar, Slotine, and
  Pavone]{singh2019robust}
S.~Singh, B.~Landry, A.~Majumdar, J.-J. Slotine, and M.~Pavone.
\newblock Robust feedback motion planning via contraction theory.
\newblock \emph{The International Journal of Robotics Research}, 2019.
\newblock Submitted.

\bibitem[Soloperto et~al.(2018)Soloperto, M{\"u}ller, Trimpe, and
  Allg{\"o}wer]{soloperto2018learning}
R.~Soloperto, M.~A. M{\"u}ller, S.~Trimpe, and F.~Allg{\"o}wer.
\newblock Learning-based robust model predictive control with state-dependent
  uncertainty.
\newblock \emph{IFAC-PapersOnLine}, 51\penalty0 (20):\penalty0 442--447, 2018.

\bibitem[Srinivas et~al.(2012)Srinivas, Krause, Kakade, and
  Seeger]{srinivas2012information}
N.~Srinivas, A.~Krause, S.~M. Kakade, and M.~W. Seeger.
\newblock Information-theoretic regret bounds for {G}aussian process
  optimization in the bandit setting.
\newblock \emph{IEEE Transactions on Information Theory}, 58\penalty0
  (5):\penalty0 3250--3265, 2012.

\bibitem[Tassa et~al.(2012)Tassa, Erez, and Todorov]{tassa2012synthesis}
Y.~Tassa, T.~Erez, and E.~Todorov.
\newblock Synthesis and stabilization of complex behaviors through online
  trajectory optimization.
\newblock In \emph{Proceedings of IEEE/RSJ International Conference on
  Intelligent Robots and Systems}, pages 4906--4913, Vilamoura, Portugal, 2012.

\bibitem[Urquhart et~al.(2020)Urquhart, Ljungskog, and
  Sebben]{urquhart2020surrogate}
M.~Urquhart, E.~Ljungskog, and S.~Sebben.
\newblock Surrogate-based optimisation using adaptively scaled radial basis
  functions.
\newblock \emph{Applied Soft Computing}, 88, 2020.
\newblock ISSN 1568-4946.
\newblock \doi{10.1016/j.asoc.2019.106050}.

\bibitem[van~der Plas and Bochenski(2020)]{fons2020pluto}
F.~van~der Plas and M.~Bochenski.
\newblock Pluto.jl, Aug. 2020.
\newblock URL \url{https://github.com/fonsp/Pluto.jl}.

\bibitem[Van Der~Vaart and Van~Zanten(2011)]{van2011information}
A.~Van Der~Vaart and H.~Van~Zanten.
\newblock Information rates of nonparametric {G}aussian process methods.
\newblock \emph{Journal of Machine Learning Research}, 12\penalty0 (6), 2011.

\bibitem[Wabersich and Zeilinger(2018)]{wabersich2018linear}
K.~P. Wabersich and M.~N. Zeilinger.
\newblock Linear model predictive safety certification for learning-based
  control.
\newblock In \emph{Proceedings of 57th IEEE Conference on Decision and
  Control}, pages 7130--7135, Miami Beach, FL, USA, 2018.

\bibitem[Wang and Hovakimyan(2012)]{wang2012l1}
X.~Wang and N.~Hovakimyan.
\newblock $\mathcal{L}_1$ adaptive controller for nonlinear time-varying
  reference systems.
\newblock \emph{Systems \& Control Letters}, 61\penalty0 (4):\penalty0
  455--463, 2012.

\bibitem[Wang et~al.(2017)Wang, Yang, Sun, and Deng]{wang2017adaptive}
X.~Wang, L.~Yang, Y.~Sun, and K.~Deng.
\newblock Adaptive model predictive control of nonlinear systems with
  state-dependent uncertainties.
\newblock \emph{International Journal of Robust and Nonlinear Control},
  27\penalty0 (17):\penalty0 4138--4153, 2017.

\bibitem[Williams and Rasmussen(2006)]{williams2006gaussian}
C.~K. Williams and C.~E. Rasmussen.
\newblock \emph{Gaussian processes for machine learning}.
\newblock MIT press, Cambridge, MA, USA, 2006.

\bibitem[Williams et~al.(2018)Williams, Drews, Goldfain, Rehg, and
  Theodorou]{williams2018information}
G.~Williams, P.~Drews, B.~Goldfain, J.~M. Rehg, and E.~A. Theodorou.
\newblock Information-theoretic model predictive control: Theory and
  applications to autonomous driving.
\newblock \emph{IEEE Transactions on Robotics}, 34\penalty0 (6):\penalty0
  1603--1622, 2018.

\bibitem[Xie et~al.(2015)Xie, van~den Berg, Patil, and Abbeel]{xie2015toward}
C.~Xie, J.~van~den Berg, S.~Patil, and P.~Abbeel.
\newblock Toward asymptotically optimal motion planning for kinodynamic systems
  using a two-point boundary value problem solver.
\newblock In \emph{Proceedings of IEEE International Conference on Robotics and
  Automation}, pages 4187--4194, Seattle, WA, USA, 2015.

\bibitem[Xu et~al.(2015)Xu, Tabuada, Grizzle, and Ames]{xu2015robustness}
X.~Xu, P.~Tabuada, J.~W. Grizzle, and A.~D. Ames.
\newblock Robustness of control barrier functions for safety critical control.
\newblock \emph{IFAC-PapersOnLine}, 48\penalty0 (27):\penalty0 54--61, 2015.

\end{thebibliography}



\appendix

\section{Proof of Theorem~\ref{thm:uniform_bounds}}\label{app:uniform_bounds}

In order to prove Theorem~\ref{thm:uniform_bounds}, we first establish the bounds on sets of finite cardinality, followed by the Lipschitz continuity of the mean functions and then the modulus of continuity of the variance functions, presented in the following lemmas.

\begin{lemma}[Bounds on sets of finite cardinality]\label{lem:finite_cardinality}
Consider the finite cardinality set $|\mathcal{Z}_\tau|$ in Definition~\ref{def:uniform_bounds} for some $\tau > 0$ and the posterior distributions in~\eqref{eqn:function:posterior}-\eqref{eqn:function_derivative:individual_posterior}. Given any $\delta \in (0,1)$, define

\[
\hat{\beta}(\tau) = 2 \log \left( \frac{m |\mathcal{Z}_\tau|}{\delta} \right), \quad
\hat{\beta}_\xi(\tau) = 2 \log \left( \frac{lm |\mathcal{Z}_\tau|}{\hat{\delta}} \right), \quad
\hat{\beta}_x(\tau) = 2 \log \left( \frac{nm |\mathcal{Z}_\tau|}{\hat{\delta}} \right),
\]
where $\hat{\delta} = 1 - (1-\delta)^{\frac{1}{m}}$.
Then, we have
\begin{align*}
    &\Pr \left\{ \norm{h(z) - \nu_N(z)} \leq \sqrt{\hat{\beta}(\tau)} \norm{\sigma_N(z)}, \quad \forall z \in \mathcal{Z}_\tau  \right\} \geq 1 - \delta,\\
    &\Pr\left\{
\norm{\left(\nabla_\xi [h]_i (z) - \nabla_\xi \nu_{i,N}(z)\right)^\top} \leq \sqrt{\hat{\beta}_\xi(\tau)} \norm{\Sigma_{\xi}^{i,N}(z)}, \quad \forall (i,z) \in \{1,\dots,m\} \times \mathcal{Z}_\tau
\right\} \geq 1 - \hat{\delta},\\
&\Pr\left\{
\norm{\left(\nabla_x [h]_i (z) - \nabla_x \nu_{i,N}(z)\right)^\top} \leq \sqrt{\hat{\beta}_\xi(\tau)} \norm{\Sigma_{x}^{i,N}(z)}, \quad \forall (i,z) \in \{1,\dots,m\} \times \mathcal{Z}_\tau
\right\} \geq 1 - \hat{\delta},
\end{align*} where $\nu_N(z)$, $\sigma_N(z)$, $\Sigma_{\xi}^{i,N}(z)$, and $\Sigma_{x}^{i,N}(z)$ are presented in Definition~\ref{def:uniform_bounds}.
\end{lemma}
\begin{proof}
The proof follows the same line of reasoning as in~\cite[Lemma~5.1]{srinivas2012information}.
Using the posterior distribution in~\cref{eqn:function:posterior}, we have that
\[
  \frac{[h]_i(z) - \nu_{i,N}(z)}{\sigma_{i,N}(z)}  \sim \mathcal{N}(0,1), \quad \forall (z,i) \in \mathcal{Z}_\tau \times \{1,\dots,m\}.
\]
Since for any $r \sim \mathcal{N}(0,1)$ and $c>0$, we have that $\Pr \left\{ |r| > c  \right\} \leq e^{-c^2/2}$.
Then, with $r = ([h]_i(z) - \nu_{i,N}(z))/\sigma_{i,N}(z)$ and $c = \sqrt{\hat{\beta}(\tau)}$, we get that
\[
\Pr \left\{ \left| \frac{[h]_i(z) - \nu_{i,N}(z)}{\sigma_{i,N}(z)} \right| > \sqrt{\hat{\beta}(\tau)}  \right\} \leq e^{-\hat{\beta}(\tau)/2}, \quad \forall (z,i) \in \mathcal{Z}_\tau \times \{1,\dots,m\}.
\]
Thus, for all $(z,i) \in \mathcal{Z}_\tau \times \{1,\dots,m\}$,
\begin{equation}\label{eqn:function:finite_cardinality:1}
 \Pr \left\{ \left| [h]_i(z) - \nu_{i,N}(z) \right| > \sqrt{\hat{\beta}(\tau)}\sigma_{i,N}(z)  \right\} \leq e^{-\hat{\beta}(\tau)/2}.
\end{equation}
Next, let $\mathcal{Z}_\tau \times \{1,\dots,m\} = \bigcup_{k} w_k$, where $k = \{1,\dots,m|\mathcal{Z}_\tau|\}$. Note that each $w_k$ is a pair of the form $(z,i)$, where $z \in \mathcal{Z}_\tau$ and $i \in \{1,\dots,m\}$. Let us define events $A_k$ as
\[
A_k = | [h]_i(z) - \nu_{i,N}(z) |  > \sqrt{\hat{\beta}(\tau)} \sigma_{i,N}(z), \quad (z,i) = w_k.
\]
Then, from~\cref{eqn:function:finite_cardinality:1} we have that $\Pr\{A_k\} \leq e^{-\hat{\beta}(\tau)/2}$, for all $k \in \{1,\dots,m|\mathcal{Z}_\tau|\}$. Applying the union bound (Boole's inequality), we get
\[
\Pr \left\{\bigcup_k A_k  \right\} \leq \sum_{k = 1}^{m|\mathcal{Z}_\tau|} \Pr \{A_k\} \leq m|\mathcal{Z}_\tau| e^{-\hat{\beta}(\tau)/2}.
\] Taking the complement, we obtain that
\begin{align*}
 \Pr \left\{ | [h]_i(z) - \nu_{i,N}(z) | \leq   \sqrt{\hat{\beta}(\tau)} \sigma_{i,N}(z) , \quad  \forall (z,i) \in \mathcal{Z}_\tau \times \{1,\dots,m\}\right\}
\geq 1 - m|\mathcal{Z}_\tau| e^{-\hat{\beta}(\tau)/2}.
\end{align*}
Therefore, we conclude that with probability at least $1 - m|\mathcal{Z}_\tau| e^{-\hat{\beta}(\tau)/2}$ we have
\begin{align*}
    \norm{h(z) - \nu_N(z)} =& \sqrt{ \sum_{i=1}^m \left|[h]_i(z) - \nu_{i,N}(z)   \right|^2 } \leq  \sqrt{ \sum_{i=1}^m \hat{\beta}(\tau) \sigma_{i,N}^2(z) }
    = \sqrt{\hat{\beta}(\tau)} \norm{\sigma_N(z)}, \quad \forall z \in \mathcal{Z}_\tau.
\end{align*}
Using the definition that $\hat{\beta}(\tau) = 2 \log \left( \frac{m|\mathcal{Z}_\tau|}{\delta} \right)$, we conclude that
\[
\Pr \left\{\norm{h(z) - \nu_N(z)} \leq  \sqrt{\hat{\beta}(\tau)} \norm{\sigma_{N}(z)}, \quad \forall z \in \mathcal{Z}_\tau  \right\} \geq 1 - \delta.
\]
Following the same line of reasoning, using the posterior distributions in~\eqref{eqn:function_derivative:posterior}-\eqref{eqn:function_derivative:individual_posterior}, we obtain that
\[
\Pr \left\{
\left|[\nabla_\xi h]_{i,k}(z) - [\nabla_\xi \nu_{i,N}]_k(z)\right| > \sqrt{\hat{\beta}_\xi(\tau)}[\nabla_\xi \sigma_{i,N}]_{k,k}(z)
\right\} \leq e^{-
\hat{\beta}_\xi(\tau)/2},
\]for all $(z,i,k) \in \mathcal{Z}_\tau \times \{1,\dots,m\} \times \{1,\dots,l\}$.

Since $|\mathcal{Z}_\tau \times \{1,\dots,m\} \times \{1,\dots,l\}| = l m |\mathcal{Z}_\tau|$ (cardinality), applying the union bound and taking the complement produces
\begin{align*}
&\Pr \left\{
\left|[\nabla_\xi h]_{i,k}(z) - [\nabla_\xi \nu_{i,N}]_k(z) \right| \leq  \sqrt{\hat{\beta}_\xi(\tau)}[\nabla_\xi \sigma_{i,N}]_{k,k}(z), \quad \forall
(z,i,k) \in \mathcal{Z}_\tau \times \{1,\dots,m\} \times \{1,\dots,l\}\right\} \\
&\geq  1 - lm |\mathcal{Z}_\tau|e^{-\hat{\beta}_\xi(\tau)/2}.
\end{align*}
Therefore, using the definition of the vector 2-norm, we get
\begin{align*}
\Pr\left\{
\norm{\left(\nabla_\xi [h]_i (z) - \nabla_\xi \nu_{i,N}(z)\right)^\top} \leq \sqrt{\hat{\beta}_\xi(\tau)} \norm{\Sigma_{\xi}^{i,N}(z)}, \quad \forall (i,z) \in \{1,\dots,m\} \times \mathcal{Z}_\tau
\right\} \geq 1 - \hat{\delta},
\end{align*} where we have used the definition of $\hat{\beta}_\xi(\tau)$. The proof for $\nabla_x h$ follows similarly.

\end{proof}


We now prove the Lipschitz continuity of the mean functions.
\begin{lemma}[Lipschitz continuity of mean functions]\label{lem:lipschitz}
Consider the posterior distributions in~\eqref{eqn:function:posterior}-\eqref{eqn:function_derivative:individual_posterior}. Then,
\begin{align*}
    \norm{\nu_{N}(z) - \nu_N(z')} \leq  L_{\nu_N} \norm{z - z'},&\\
    \norm{\left(\nabla_\xi \nu_{i,N}(z) - \nabla_\xi \nu_{i,N}(z')  \right)^\top  } \leq \nabla_\xi L_{i,\nu_N}  \norm{z - z'},&\\
     \norm{\left(\nabla_x \nu_{i,N}(z) - \nabla_x \nu_{i,N}(z')  \right)^\top  } \leq  \nabla_x L_{i,\nu_N}\norm{z - z'},&
\end{align*} for all $z,z' \in \mathcal{Z}$ and $i \in \{1,\dots,m\}$, where $L_{\nu_N}$, $\nabla_\xi L_{i,\nu_N}$, and $\nabla_x L_{i,\nu_N}$ are defined in the statement of Theorem~\ref{thm:uniform_bounds}.
\end{lemma}
\begin{proof}

Using the definition of $\nu_{i,N}(z)$ in~\eqref{eqn:function:posterior}, we get
\[
\left| \nu_{i,N}(z) - \nu_{i,N}(z') \right| = \left|
\left( K_i(z,\mathbf{Z}) - K_i(z',\mathbf{Z})  \right)^\top
\left[K_i(\mathbf{Z},\mathbf{Z}) + \sigma^2 \mathbb{I}_N  \right]^{-1} \left([\mathbf{Y}]_{i,\cdot}\right)^\top
\right|, \quad \forall z,z' \in \mathcal{Z},~i \in \{1,\dots,m\}.
\]
Applying the Cauchy-Schwarz inequality, we get
\begin{equation}\label{eqn:function:Lipschitz_continuity:1}
    \left| \nu_{i,N}(z)  - \nu_{i,N}(z') \right| \le
\norm{ K_i(z,\mathbf{Z}) - K_i(z',\mathbf{Z}) }
\norm{ \left[K_i(\mathbf{Z},\mathbf{Z}) + \sigma^2 \mathbb{I}_N  \right]^{-1} \left([\mathbf{Y}]_{i,\cdot}\right)^\top },
\end{equation}
for all $z,z' \in \mathcal{Z}$ and $i \in \{1,\dots,m\}$. From the definition of $K_i(z,\mathbf{Z})$ in~\cref{eqn:function:posterior}, we get that
\[
\norm{ K_i(z,\mathbf{Z}) - K_i(z',\mathbf{Z}) }^2
=
 \sum_{j=1}^N \left( K_i(z,z_j) - K_i(z',z_j)   \right)^2  .
\] Using the Lipschitz continuity of the kernel functions in Assumption~\ref{assmp:GP}, we further obtain
\begin{align}
    \norm{ K_i(z,\mathbf{Z}) - K_i(z',\mathbf{Z}) }^2
=   \sum_{j=1}^N \left( K_i(z,z_j) - K_i(z',z_j)   \right)^2  \leq &  \sum_{j=1}^N L_{K_i}^2 \norm{z - z'}^2 \notag   \\
= &\label{eqn:function:Lipschitz_continuity:2} N L_{K_i}^2 \norm{z - z'}^2, \quad \forall z,z' \in \mathcal{Z},~i \in \{1,\dots,m\}.
\end{align}
Using this inequality with~\cref{eqn:function:Lipschitz_continuity:1} produces
\[
\left| \nu_{i,N}(z) - \nu_{i,N}(z')  \right|^2 \leq
 N L_{K_i}^2
\norm{ \left[K_i(\mathbf{Z},\mathbf{Z}) + \sigma^2 \mathbb{I}_N  \right]^{-1} \left([\mathbf{Y}]_{i,\cdot}\right)^\top }^2  \norm{z - z'}^2,
\] for all $z,z' \in \mathcal{Z}$ and $i \in \{1,\dots,m\}$.
Therefore, we obtain
\begin{align*}
    \norm{\nu_{N}(z) - \nu_N(z')} = & \sqrt{ \sum_{i=1}^m \left| \nu_{i,N}(z) - \nu_{i,N}(z') \right|^2     } \\
    \leq & \sqrt{ \sum_{i=1}^m N L_{K_i}^2
\norm{ \left[K_i(\mathbf{Z},\mathbf{Z}) + \sigma^2 \mathbb{I}_N  \right]^{-1} \left([\mathbf{Y}]_{i,\cdot}\right)^\top }^2  \norm{z - z'}^2     } \\
=& \sqrt{ N \left( \sum_{i=1}^m  L_{K_i}^2
\norm{ \left[K_i(\mathbf{Z},\mathbf{Z}) + \sigma^2 \mathbb{I}_N  \right]^{-1} \left([\mathbf{Y}]_{i,\cdot}\right)^\top }^2 \right)}\norm{z - z'} \\
= & L_{\nu_N} \norm{z - z'}, \quad \forall z,z' \in \mathcal{Z}.
\end{align*}

Continuing, using the definition of $\nabla_\xi \nu_{i,N}(z)$ in~\cref{eqn:function_derivative:posterior}, we get that
\begin{align}
    \norm{\left(\nabla_\xi \nu_{i,N}(z) - \nabla_\xi \nu_{i,N}(z')  \right)^\top  } &\label{lem:function_derivative:Lip:1} \leq \norm{ \nabla_\xi K_i(z,\mathbf{Z}) - \nabla_\xi K_i(z',\mathbf{Z})}\norm{\left[ K_i(\mathbf{Z},\mathbf{Z}) + \sigma^2 \mathbb{I}_N\right]^{-1}\left([\mathbf{Y}]_{i,\cdot}\right)^\top},
\end{align} for all $i \in \{1,\dots,m\}$ and $z,z' \in \mathcal{Z}$.

Then we have
\begin{align*}
    \norm{\nabla_\xi K_i(z,\mathbf{Z}) - \nabla_\xi K_i(z',\mathbf{Z}) }_F =&
    \sqrt{ \sum_{j=1}^N \sum_{k = 1}^l \left| \pdv{K_i}{\xi_k}(z,z_j) - \pdv{K_i}{\xi_k}(z',z_j)   \right|^2   } \\
    =& \sqrt{ \sum_{j=1}^N \norm{\left( \nabla_\xi K_i(z,z_j) - \nabla_\xi K_i(z',z_j)    \right)^\top  }^2           },
\end{align*}
where, $\norm{\cdot}_F$ denotes the Frobenius norm.
Using~\cref{assmp:GP}, we obtain
\begin{align*}
    \norm{\nabla_\xi K_i(z,\mathbf{Z}) - \nabla_\xi K_i(z',\mathbf{Z})  }_F \leq &
    \sqrt{ N \left( \nabla_\xi L_{K_i} \right)^2   \norm{z - z'}^2} =\sqrt{N} \nabla_\xi L_{K_i} \norm{z - z'},
\end{align*} for all $i \in \{1,\dots,m\}$ and $z,z' \in \mathcal{Z}$. Since $\norm{\cdot} \leq \norm{\cdot}_F$, substituting the aforementioned expression into~\cref{lem:function_derivative:Lip:1} produces the desired result. The Lipschitz continuity of $\nabla_x \nu_{i,N}(z)$ is established similarly.

\end{proof}

Next, we proceed towards deriving the modulus of continuity of the variance functions.
\begin{lemma}[Modulus of continuity of variance functions]\label{lem:modulus}
Consider the posterior distributions in~\eqref{eqn:function:posterior}-\eqref{eqn:function_derivative:individual_posterior}. Then,
\begin{align*}
    \norm{\sigma_N(z) - \sigma_N(z')} \leq \omega_N\left(\norm{z - z'}\right),&\\
    \norm{ \Sigma_{\xi}^{i,N}(z) - \Sigma_{\xi}^{i,N}(z')  } \leq \nabla_\xi \omega_{i,N} \left( \norm{z - z'} \right),&\\
\norm{ \Sigma_{x}^{i,N}(z) - \Sigma_{x}^{i,N}(z')  } \leq \nabla_x \omega_{i,N} \left( \norm{z - z'} \right),&
\end{align*} for all $z,z' \in \mathcal{Z}$ and $i \in \{1,\dots,m\}$, where $\sigma_N(z)$, $\Sigma_{\xi}^{i,N}(z)$, and $\Sigma_x^{i,N}(z)$ are presented in Definition~\ref{def:uniform_bounds}, and $\omega_N(\cdot)$, $\nabla_\xi \omega_{i,N}(\cdot)$, and $\nabla_x \omega_{i,N}(\cdot)$ are defined in the statement of Theorem~\ref{thm:uniform_bounds}.
\end{lemma}
\begin{proof}

Using the positivity of the variance functions, we get
\begin{equation}\label{eqn:function:modulus_continuity:1}
    \left|\sigma_{i,N}^2(z) - \sigma_{i,N}^2(z') \right| \geq \left|\sigma_{i,N}(z) - \sigma_{i,N}(z') \right|^2, \quad \forall z,z' \in \mathcal{Z},~i \in \{1,\dots,m\}.
\end{equation}
Using the definition of $\sigma_{i,N}(z)$ in~\cref{eqn:function:posterior}, we upper bound
\begin{align}
    \left|\sigma_{i,N}^2(z) - \sigma_{i,N}^2(z') \right|  \leq &
    \left| K_i(z,z) - K_i(z',z')\right|  \notag \\
    &\label{eqn:function:modulus_continuity:2}+ \norm{K_i(z,\mathbf{Z}) - K_i(z',\mathbf{Z})} \norm{\left[K_i(\mathbf{Z},\mathbf{Z})
    + \sigma^2 \mathbb{I}_N  \right]^{-1}} \norm{K_i(z,\mathbf{Z}) + K_i(z',\mathbf{Z})},
\end{align}
for all $z,z' \in \mathcal{Z}$ and $i \in \{1,\dots,m\}$, where we have used the fact that
\begin{align*}
&K_i(z,\mathbf{Z})^\top \left[K_i(\mathbf{Z},\mathbf{Z}) + \sigma^2 \mathbb{I}_N  \right]^{-1}  K_i(z,\mathbf{Z})    - K_i(z',\mathbf{Z})^\top \left[K_i(\mathbf{Z},\mathbf{Z}) + \sigma^2 \mathbb{I}_N  \right]^{-1}  K_i(z',\mathbf{Z}) \\
& = \left(K_i(z,\mathbf{Z}) - K_i(z',\mathbf{Z})  \right)^\top \left[K_i(\mathbf{Z},\mathbf{Z}) + \sigma^2 \mathbb{I}_N  \right]^{-1}
\left(K_i(z,\mathbf{Z}) + K_i(z',\mathbf{Z})  \right).
\end{align*}
Next, using the Lipschitz continuity of the kernel functions in~\cref{assmp:GP} and their symmetry in the arguments, we get that
\begin{align}
    \left| K_i(z,z) - K_i(z',z')\right| = & \left| K_i(z,z) - K_i(z,z') + K_i(z,z') - K_i(z',z')\right| \notag \\
    \leq & \left| K_i(z,z) - K_i(z',z) \right| + \left| K_i(z,z') - K_i(z',z') \right| \notag \\
    = &\label{eqn:function:modulus_continuity:3} 2 L_{K_i} \norm{z - z'}, \quad \forall z,z' \in \mathcal{Z},~i \in \{1,\dots,m\}.
\end{align}
Further, using~\cref{eqn:function:Lipschitz_continuity:2} from the proof of~\cref{lem:lipschitz}, we get
\begin{equation}\label{eqn:function:modulus_continuity:4}
   \norm{ K_i(z,\mathbf{Z}) - K_i(z',\mathbf{Z}) } \leq \sqrt{N} L_{K_i} \norm{z - z'}, \quad \forall z,z' \in \mathcal{Z},~i \in \{1,\dots,m\}.
\end{equation}
Moreover, we have the following identity
\begin{equation}\label{eqn:function:modulus_continuity:5}
    \norm{K_i(z,\mathbf{Z}) + K_i(z',\mathbf{Z})} \leq 2 \sqrt{N} \max_{z,z' \in \mathcal{Z}} K_i(z,z'), \quad \forall z,z' \in \mathcal{Z},~i \in \{1,\dots,m\}.
\end{equation}
Substituting~\cref{eqn:function:modulus_continuity:3}-\cref{eqn:function:modulus_continuity:5} into~\cref{eqn:function:modulus_continuity:2}, we get
\begin{align*}
    &\left|\sigma_{i,N}^2(z) - \sigma_{i,N}^2(z') \right| \leq  2 L_{K_i} \norm{z - z'} \left(1 + N \norm{\left[K_i(\mathbf{Z},\mathbf{Z}) + \sigma^2 \mathbb{I}_N  \right]^{-1}} \max_{z,z' \in \mathcal{Z}} K_i(z,z')  \right),
\end{align*}
for all $z,z' \in \mathcal{Z}$, $i \in \{1,\dots,m\}$. Then, using~\cref{eqn:function:modulus_continuity:1}, we get
\begin{align*}
  &\left|\sigma_{i,N}(z) - \sigma_{i,N}(z') \right|^2  \leq 2 L_{K_i} \norm{z - z'} \left(1 + N \norm{\left[K_i(\mathbf{Z},\mathbf{Z}) + \sigma^2 \mathbb{I}_N  \right]^{-1}} \max_{z,z' \in \mathcal{Z}} K_i(z,z')  \right),
\end{align*}
for all $z,z' \in \mathcal{Z}$, $i \in \{1,\dots,m\}$. Therefore
\begin{align*}
    \norm{\sigma_N(z) - \sigma_N(z')} &= \sqrt{\sum_{i=1}^m \left|\sigma_{i,N}(z) - \sigma_{i,N}(z') \right|^2  } \\
    &\leq  \sqrt{ 2 \norm{z - z'} \sum_{i=1}^m   L_{K_i}  \left(1 + N \norm{\left[K_i(\mathbf{Z},\mathbf{Z}) + \sigma^2 \mathbb{I}_N  \right]^{-1}} \max_{z,z' \in \mathcal{Z}} K_i(z,z')  \right) }\\
    &= \omega_N\left(\norm{z - z'}\right), \quad \forall z,z' \in \mathcal{Z}.
\end{align*}
Continuing on for the variance of the partial derivatives of the uncertainty, from~\eqref{eqn:function_derivative:posterior}-\eqref{eqn:function_derivative:individual_posterior} we have
\begin{align*}
    \left[ \nabla_\xi \sigma_{i,N}^2 \right]_{k,k}(z) = & \frac{\partial^2 K_i}{\partial \xi_k \partial \xi'_k}(z,z) -
    \begin{bmatrix}
    \pdv{K_i}{\xi_k}(z,z_1) \\ \vdots \\ \pdv{K_i}{\xi_k}(z,z_N)
    \end{bmatrix}^\top
    \left[ K_i(\mathbf{Z},\mathbf{Z}) + \sigma^2 \mathbb{I}_N\right]^{-1}
    \begin{bmatrix}
    \pdv{K_i}{\xi_k}(z,z_1) \\ \vdots \\ \pdv{K_i}{\xi_k}(z,z_N)
    \end{bmatrix},
\end{align*} for all $(i,k) \in \{1,\dots,m\} \times \{1,\dots,l\}$ and $z \in \mathcal{Z}$.
Following the same line of reasoning as for the posterior variance of the uncertainty, we have that
\begin{equation}\label{eqn:function_derivative:mod_cont:1}
\left| \left[ \nabla_\xi \sigma_{i,N}^2 \right]_{k,k}(z) - \left[ \nabla_\xi \sigma_{i,N}^2 \right]_{k,k}(z') \right|
\geq
\left| \left[ \nabla_\xi \sigma_{i,N} \right]_{k,k}(z) - \left[ \nabla_\xi \sigma_{i,N} \right]_{k,k}(z') \right|^2,
\end{equation}
for all $(i,k) \in \{1,\dots,m\} \times \{1,\dots,l\}$ and $z,z' \in \mathcal{Z}$.
Thus, we have
\begin{align}
   &\left| \left[ \nabla_\xi \sigma_{i,N}^2 \right]_{k,k}(z) - \left[ \nabla_\xi \sigma_{i,N}^2 \right]_{k,k}(z') \right|
     \\
    & \leq \left| \frac{\partial^2 K_i}{\partial \xi_k \partial \xi'_k}(z,z) - \frac{\partial^2 K_i}{\partial \xi_k \partial \xi'_k}(z',z') \right|  +
  \norm{\begin{bmatrix}
    \pdv{K_i}{\xi_k}(z,z_1) - \pdv{K_i}{\xi_k}(z',z_1) \\ \vdots \\ \pdv{K_i}{\xi_k}(z,z_N) - \pdv{K_i}{\xi_k}(z',z_N)
    \end{bmatrix}}
    \norm{\left[ K_i(\mathbf{Z},\mathbf{Z}) + \sigma^2 \mathbb{I}_N\right]^{-1}} \notag \\
    &\label{eqn:function_derivative:mod_cont:2}\qquad \qquad \times \norm{\begin{bmatrix}
    \pdv{K_i}{\xi_k}(z,z_1) + \pdv{K_i}{\xi_k}(z',z_1) \\ \vdots \\ \pdv{K_i}{\xi_k}(z,z_N) + \pdv{K_i}{\xi_k}(z',z_N)
    \end{bmatrix}},
\end{align} for all $(i,k) \in \{1,\dots,m\} \times \{1,\dots,l\}$ and $z,z' \in \mathcal{Z}$.

As before, we now proceed by bounding the terms on the right hand side using Assumption~\ref{assmp:GP}. We start with
\begin{align}
 \left| \frac{\partial^2 K_i}{\partial \xi_k \partial \xi'_k}(z,z) - \frac{\partial^2 K_i}{\partial \xi_k \partial \xi'_k}(z',z') \right|
 & \leq \left| \frac{\partial^2 K_i}{\partial \xi_k \partial \xi'_k}(z,z) - \frac{\partial^2 K_i}{\partial \xi_k \partial \xi'_k}(z,z') \right| +
 \left| \frac{\partial^2 K_i}{\partial \xi_k \partial \xi'_k}(z,z') - \frac{\partial^2 K_i}{\partial \xi_k \partial \xi'_k}(z',z') \right| \notag \\
 &\label{eqn:function_derivative:mod_cont:3} \leq 2 \nabla_\xi L_{K_i} \norm{z - z'}, \quad \forall (i,k) \in \{1,\dots,m\} \times \{1,\dots,l\},~z,z' \in \mathcal{Z}.
\end{align}
Proceeding further, we have
\begin{align}
    \norm{\begin{bmatrix}
    \pdv{K_i}{\xi_k}(z,z_1) - \pdv{K_i}{\xi_k}(z',z_1) \\ \vdots \\ \pdv{K_i}{\xi_k}(z,z_N) - \pdv{K_i}{\xi_k}(z',z_N)
    \end{bmatrix}} \leq
    \sqrt{ \sum_{j=1}^N \left| \pdv{K_i}{\xi_k}(z,z_j) - \pdv{K_i}{\xi_k}(z',z_j) \right|^2  } \leq &
    \sqrt{ \sum_{j=1}^N \sum_{k=1}^l \left| \pdv{K_i}{\xi_k}(z,z_j) - \pdv{K_i}{\xi_k}(z',z_j) \right|^2  } \notag \\
    \leq &\label{eqn:function_derivative:mod_cont:4} \sqrt{N}\nabla_\xi L_{K_i} \norm{z - z'},
\end{align}
for all $(i,k) \in \{1,\dots,m\} \times \{1,\dots,l\}$ and $z,z' \in \mathcal{Z}$, and where we have used the computation used in~\cref{lem:lipschitz}.
Finally,
\begin{align}
    \norm{\begin{bmatrix}
    \pdv{K_i}{\xi_k}(z,z_1) + \pdv{K_i}{\xi_k}(z',z_1) \\ \vdots \\ \pdv{K_i}{\xi_k}(z,z_N) + \pdv{K_i}{\xi_k}(z',z_N)
    \end{bmatrix}} &\leq  \sqrt{  \sum_{j=1}^N \left| \pdv{K_i}{\xi_k}(z,z_j) + \pdv{K_i}{\xi_k}(z',z_j)  \right|^2 } \notag \\
    &\label{eqn:function_derivative:mod_cont:5} \leq 2 \sqrt{N} \max_{z,z' \in \mathcal{Z}}\left| \pdv{K_i}{\xi_k}(z,z') \right|.
\end{align}
Substituting~\eqref{eqn:function_derivative:mod_cont:3}-\eqref{eqn:function_derivative:mod_cont:5} into~\eqref{eqn:function_derivative:mod_cont:2} produces
\begin{align*}
     &\left| \left[ \nabla_\xi \sigma_{i,N}^2 \right]_{k,k}(z) - \left[ \nabla_\xi \sigma_{i,N}^2 \right]_{k,k}(z') \right| \\
     & \leq 2 \nabla_\xi L_{K_i} \norm{z - z'}  + \sqrt{N} \nabla_\xi L_{K_i} \norm{z - z'} \norm{\left[ K_i(\mathbf{Z},\mathbf{Z}) + \sigma^2 \mathbb{I}_N\right]^{-1}} 2 \sqrt{N}\max_{z,z' \in \mathcal{Z}}\left| \pdv{K_i}{\xi_k}(z,z') \right|,
\end{align*}
for all $(i,k) \in \{1,\dots,m\} \times \{1,\dots,l\}$ and $z,z' \in \mathcal{Z}$.
Thus, from~\eqref{eqn:function_derivative:mod_cont:1}, we get
\begin{align}
     &\left| \left[ \nabla_\xi \sigma_{i,N} \right]_{k,k}(z) - \left[ \nabla_\xi \sigma_{i,N} \right]_{k,k}(z') \right|^2 \notag \\
&\label{eqn:function_derivative:mod_cont:6} \leq 2 \nabla_\xi L_{K_i} \norm{z - z'} \left(1 + N \norm{\left[ K_i(\mathbf{Z},\mathbf{Z}) + \sigma^2 \mathbb{I}_N\right]^{-1}}  \max_{z,z' \in \mathcal{Z}}\left| \pdv{K_i}{\xi_k}(z,z') \right| \right),
\end{align}
for all $(i,k) \in \{1,\dots,m\} \times \{1,\dots,l\}$ and $z,z' \in \mathcal{Z}$.
Next, from~\cref{def:uniform_bounds}, we have
\begin{align*}
    \norm{ \Sigma_\xi^{i,N}(z) - \Sigma_\xi^{i,N}(z')  } = \sqrt{ \sum_{k = 1}^l  \left| \left[ \nabla_\xi \sigma_{i,N} \right]_{k,k}(z) - \left[ \nabla_\xi \sigma_{i,N} \right]_{k,k}(z') \right|^2 }.
\end{align*}
Then, using~\cref{eqn:function_derivative:mod_cont:6} gives us the modulus of continuity for $\Sigma_\xi^{i,N}(z)$. The proof for the modulus of continuity of $\Sigma_x^{i,N}(z)$ follows similarly.
\end{proof}

We now use the results in Lemmas~\ref{lem:finite_cardinality}-\ref{lem:modulus} to prove Theorem~\ref{thm:uniform_bounds}.
\begin{proof}[Proof of Theorem~\ref{thm:uniform_bounds}]

Using Assumptions~\ref{assmp:models:functions} and~\ref{assmp:models:functions_derivatives}, it is straightforward to establish that
\begin{equation}\label{eqn:thm_function1a}
    \norm{h(z) - h(z')} \leq  \left(\Delta_{h_x} + \Delta_{h_\xi}  \right)\norm{z - z'}, \quad \forall z \in \mathcal{Z},~z' \in \mathcal{Z}_\tau.
\end{equation}
Furthermore, by Lemmas~\ref{lem:lipschitz} and~\ref{lem:modulus}, we have that
\begin{subequations}
\begin{align}
\norm{\nu_N(z) - \nu_N(z')} \leq &\label{eqn:thm_function1b} L_{\nu_N}\norm{z - z'},\\
\norm{\sigma_N(z) - \sigma_N(z')} \leq &\label{eqn:thm_function1c} \omega_N \left( \norm{z - z'} \right),
\end{align}
\end{subequations}
for all $z \in \mathcal{Z}$ and $z' \in \mathcal{Z}_\tau$.

We can now compute the upper bound
\[
\norm{h(z) - \nu_N(z)} \leq \norm{h(z) - h(z')} + \norm{\nu_N(z) - \nu_N(z')} + \norm{h(z') - \nu_N(z')}, \quad \forall z \in \mathcal{Z},~z' \in \mathcal{Z}_\tau.
\]
Therefore, from Lemma~\ref{lem:finite_cardinality} and Equations~\eqref{eqn:thm_function1a} and~\eqref{eqn:thm_function1b} we get
\begin{equation}\label{eqn:thm_function:2}
\norm{h(z) - \nu_{N}(z)} \leq \left(\Delta_{h_x} + \Delta_{h_\xi} + L_{\nu_N}  \right)\norm{z - z'} + \sqrt{\hat{\beta}(\tau)} \norm{\sigma_N(z')},
\end{equation} for all $z \in \mathcal{Z}$, $z' \in \mathcal{Z}_\tau$ holds w.p. at-least $1-\delta$.

Further, we have
\[
\norm{\sigma_N(z')} \leq \norm{\sigma_N(z') - \sigma_N(z)} + \norm{\sigma_N(z)}, \quad \forall z \in \mathcal{Z},~z' \in \mathcal{Z}_\tau.
\]
Thus, from~\eqref{eqn:thm_function1c} it follows
\[
\norm{\sigma_N(z')} \leq \omega_N \left(\norm{z - z'} \right) + \norm{\sigma_N(z)}.
\]
Substituting into~\cref{eqn:thm_function:2} we get that
\[
\norm{h(z) - \nu_{N}(z)} \leq \left(\Delta_{h_x} + \Delta_{h_\xi} + L_{\nu_N}  \right)\norm{z - z'} + \sqrt{\hat{\beta}(\tau)} \omega_N \left(\norm{z - z'} \right) +  \sqrt{\hat{\beta}(\tau)}\norm{\sigma_N(z)} ,
\]
for all $z \in \mathcal{Z}$, $z' \in \mathcal{Z}_\tau$ holds w.p. at-least $1-\delta$. Finally, since $\max_{z \in \mathcal{Z}} \min_{z' \in \mathcal{Z}_\tau} \norm{z - z'} \leq \tau$ and the minimum number of grid points $|\mathcal{Z}_\tau|$ is given by the covering number $M(\tau,\mathcal{Z})$, which implies that $\hat{\beta}(\tau) \leq \beta(\tau)$, we get
\[
\Pr \left\{
\norm{h(z) - \nu_{N}(z)} \leq \sqrt{\beta(\tau)}\norm{\sigma_N(z)} + \gamma(\tau) , \forall z \in \mathcal{Z}
\right\} \geq 1 - \delta.
\]


We follow similar lines for the derivative of the uncertainty. We have
\begin{align}
    \norm{\left( \nabla_\xi [h]_i(z) - \nabla_\xi \nu_{i,N}(z) \right)^\top } \leq & \norm{ \left( \nabla_\xi [h]_i(z) - \nabla_\xi [h]_i(z') \right)^\top } +
    \norm{ \left( \nabla_\xi \nu_{i,N}(z) - \nabla_\xi \nu_{i,N}(z') \right)^\top  } \notag \\
    &\label{eqn:function_derivative:uniform_bound:1} + \norm{\left( \nabla_\xi [h]_i(z') - \nabla_\xi \nu_{i,N}(z')    \right)^\top},
\end{align} for all $i \in \{1,\dots,m\}$, $z \in \mathcal{Z}$, $z' \in \mathcal{Z}_\tau$.
From~\cref{assmp:models:functions_derivatives}. we have that
\begin{equation}\label{eqn:function_derivative:uniform_bound:2}
    \norm{ \left( \nabla_\xi [h]_i(z) - \nabla_\xi [h]_i(z') \right)^\top } \leq \nabla_\xi \Delta_{h_\xi}^i \norm{z - z'}, \quad \forall i \in \{1,\dots,m\},~z \in \mathcal{Z},~z' \in \mathcal{Z}_\tau.
\end{equation}
Using~\cref{lem:lipschitz} we obtain
\begin{equation}\label{eqn:function_derivative:uniform_bound:3}
    \norm{ \left( \nabla_\xi \nu_{i,N}(z) - \nabla_\xi \nu_{i,N}(z') \right)^\top  } \leq \nabla_\xi L_{i,\nu_N}\norm{z - z'}, \quad i \in \{1,\dots,m\},~z \in \mathcal{Z},~z' \in \mathcal{Z}_\tau.
\end{equation}
From~\cref{lem:finite_cardinality} we have that
\begin{equation}\label{eqn:function_derivative:uniform_bound:4}
    \norm{\left( \nabla_\xi [h]_i(z') - \nabla_\xi \nu_{i,N}(z')    \right)^\top} \leq \sqrt{\hat{\beta}_\xi(\tau)} \norm{ \Sigma_\xi^{i,N}(z') }, \quad
    \forall i \in \{1,\dots,m\},~z' \in \mathcal{Z}_\tau
\end{equation} holds w.p. at least $1 - \hat{\delta}$.
Substituting~\eqref{eqn:function_derivative:uniform_bound:2}-\eqref{eqn:function_derivative:uniform_bound:4} into~\eqref{eqn:function_derivative:uniform_bound:1} produces the fact that
\begin{align}
    \norm{\left( \nabla_\xi [h]_i(z) - \nabla_\xi \nu_{i,N}(z) \right)^\top } \leq &\label{eqn:function_derivative:uniform_bound:5}
    \left( \nabla_\xi \Delta_{h_\xi}^i + \nabla_\xi L_{i,\nu_N}  \right)\norm{z - z'}  + \sqrt{\hat{\beta}_\xi(\tau)} \norm{ \Sigma_\xi^{i,N}(z') }
\end{align} holds for each $i \in \{1,\dots,m\}$, $z \in \mathcal{Z}$, $z' \in \mathcal{Z}_\tau$ w.p. at least $1 - \hat{\delta}$.
Thus, we have
\begin{align*}
   \sqrt{\hat{\beta}_\xi(\tau)} \norm{ \Sigma_\xi^{i,N}(z') } \leq \sqrt{\hat{\beta}_\xi(\tau)} \norm{ \Sigma_\xi^{i,N}(z') - \Sigma_\xi^{i,N}(z) } +  \sqrt{\hat{\beta}_\xi(\tau)} \norm{ \Sigma_\xi^{i,N}(z) },
\end{align*} holds for all $i \in \{1,\dots,m\}$, $z \in \mathcal{Z}$, $z' \in \mathcal{Z}_\tau$. Applying~\cref{lem:modulus}, we get
\begin{equation}\label{eqn:function_derivative:uniform_bound:6}
 \sqrt{\hat{\beta}_\xi(\tau)} \norm{ \Sigma_\xi^{i,N}(z') } \leq \sqrt{\hat{\beta}_\xi(\tau)} \nabla_\xi \omega_{i,N}\left( \norm{z - z'} \right) +  \sqrt{\hat{\beta}_\xi(\tau)} \norm{ \Sigma_\xi^{i,N}(z) },
\end{equation} for all $i \in \{1,\dots,m\}$, $z \in \mathcal{Z}$, $z' \in \mathcal{Z}_\tau$.
Substituting~\cref{eqn:function_derivative:uniform_bound:6} into~\cref{eqn:function_derivative:uniform_bound:5} leads to the conclusion that
\[
  \norm{\left( \nabla_\xi [h]_i(z) - \nabla_\xi \nu_{i,N}(z) \right)^\top } \leq \nabla_\xi \gamma_i \left( \norm{z - z'} \right) +  \sqrt{\hat{\beta}_\xi(\tau)} \norm{ \Sigma_\xi^{i,N}(z) }
\] holds for each $i \in \{1,\dots,m\}$, $z \in \mathcal{Z}$, $z' \in \mathcal{Z}_\tau$ w.p. at least $1 - \hat{\delta}$. Using the properties of the set $\mathcal{Z}_\tau$ in~\cref{def:uniform_bounds}, we get that
\[
  \norm{\left( \nabla_\xi [h]_i(z) - \nabla_\xi \nu_{i,N}(z) \right)^\top } \leq \nabla_\xi \gamma_i \left( \tau \right) +  \sqrt{\hat{\beta}_\xi(\tau)} \norm{ \Sigma_\xi^{i,N}(z) },
\] holds for each $i \in \{1,\dots,m\}$, $z \in \mathcal{Z}$ w.p. at least $1 - \hat{\delta}$. Using the independence of each $[h]_i(z) - \nu_{i,N}(z)$, we get that
\begin{align}
    &\Pr \left\{
      \norm{\left( \nabla_\xi [h]_i(z) - \nabla_\xi \nu_{i,N}(z) \right)^\top } \leq \nabla_\xi \gamma_i \left( \tau \right) +  \sqrt{\hat{\beta}_\xi(\tau)} \norm{ \Sigma_\xi^{i,N}(z) }, \quad \forall i \in \{1,\dots,m\},~z \in \mathcal{Z}
    \right\} \notag \\
    &\label{eqn:function_derivative:uniform_bound:7} \geq (1 - \hat{\delta})^m.
\end{align}

We can define
\[
\nabla_\xi h(z) = \begin{bmatrix}
\nabla_\xi [h]_1 (z) \\ \vdots \\ \nabla_\xi [h]_m (z)
\end{bmatrix} \in \mathbb{R}^{m \times l}, \quad
\nabla_\xi \nu_N(z) = \begin{bmatrix}
\nabla_\xi \nu_{1,N} (z) \\ \vdots \\ \nabla_\xi \nu_{m,N} (z)
\end{bmatrix} \in \mathbb{R}^{m \times l}.
\] Thus, we use~\cref{eqn:function_derivative:uniform_bound:7} to compute the Frobenius norm and obtain that
\begin{align*}
    \norm{\nabla_\xi h(z) - \nabla_\xi \nu_N(z)}_F =  \sqrt{ \sum_{i=1}^m \norm{ \nabla_\xi [h]_i(z) - \nabla_\xi \nu_{i,N}(z) }^2   } \leq & \sqrt{ \sum_{i=1}^m \left( \nabla_\xi \gamma_i(\tau) + \sqrt{\hat{\beta}_\xi(\tau)}\norm{\Sigma_\xi^{i,N}(z) }   \right)^2  } \\
    = &\nabla_\xi \Delta_{h} (z,\tau)
\end{align*} holds for all $z \in \mathcal{Z}$ w.p. at least $(1-\hat{\delta})^m$. We conclude the result by observing that $\norm{\cdot} \leq \norm{\cdot}_F$. Furthermore, as presented in Definition~\ref{def:uniform_bounds}, we have that $(1 - \hat{\delta})^m = 1- \delta$. Finally, as before,  $\hat{\beta}_\xi(\tau) \leq \beta_\xi(\tau)$ since the minimum number of grid points $|\mathcal{Z}_\tau|$ is given by the covering number $M(\tau,\mathcal{Z})$.
The proof for $\nabla_x h(z)$ is completed in a similar fashion.

\end{proof}

\end{document}